\documentclass[journal,12pt,onecolumn]{IEEEtran}

\usepackage{amsmath,amsfonts}
\usepackage{amssymb,verbatim,amsfonts,amsmath,amsthm}
\usepackage{geometry}
\usepackage{cite}
\usepackage{hyperref}
\usepackage{bm}

\usepackage{longtable}
\usepackage{booktabs}
\usepackage[justification=centering]{caption}
\usepackage{graphicx}
\usepackage[figuresright]{rotating}

\usepackage{seqsplit}

\newtheorem{theorem}{Theorem}
\newtheorem{lemma}[theorem]{Lemma}
\newtheorem{corollary}[theorem]{Corollary}
\newtheorem{proposition}[theorem]{Proposition}

\newtheorem{remark}[theorem]{Remark}

\newcommand{\Fq}{\mathbb{F}_q}

\begin{document}

\author{Haojie Chen
\thanks{Haojie Chen is with the School of Mathematics, Sun Yat-sen University, Guangzhou 510275,
China (e-mail: chenhj69@mail2.sysu.edu.cn).}
}

\title{A Uniform Upper Bound on the Length of MDS Codes from Genus-Two Curves}

\maketitle

\begin{abstract}
We study the maximal length of MDS algebraic geometry codes
from genus-two curves over finite fields. Let $C_\mathcal{ L}(D,G)$ be
such a code over $\mathbb F_q$, with length $n$ and dimension $k$.
Assume that $3\leq k\leq n-3$.
We prove that $n<(2q+10\sqrt q+16)/3$.
This bound holds for every prime power $q$ and for both
parities of $k$. It gives a uniform upper bound of
$2q/3+O(\sqrt q)$ for the maximal length.
The explicit bound is strictly smaller than $q+1$
when $q\geq125$.

{\bf Index Terms:}  MDS codes, Hyperelliptic curves, Algebraic geometry codes, Function fields.
\end{abstract}

\section{Introduction}
\label{sec:introduction}

Maximum distance separable (MDS) codes are optimal linear codes
with respect to the Singleton bound. A linear $[n,k,d]_q$ code over $\Fq$
has length $n$, dimension $k$, and minimum distance $d$, and is
called MDS if $d=n-k+1$. Their optimal distance properties make
MDS codes fundamental objects in coding theory; see
\cite{MR465509}. In addition to the classical questions about
their possible parameters, MDS codes with prescribed algebraic
properties have also been studied. A basic problem is to determine
how long an MDS code can be over a given finite field.

Let $M(k,q)$ denote the maximal length of a non-trivial $q$-ary
linear MDS code of dimension $k$. The
classical MDS main conjecture (see \cite{MR465509}) states that for $q>k$, 
\[
    M(k,q)=
    \begin{cases}
        q+2, & q\text{ even and }k\in\{3,q-1\},\\
        q+1, & \text{otherwise}.
    \end{cases}
\]

Much progress has been made toward this conjecture.
The connection between MDS codes and arcs in finite projective
spaces has led to substantial progress on this problem.
Ball~\cite{MR2911882} proved the conjecture over prime fields.
For odd square fields, results on planar arcs by Ball and
Lavrauw~\cite{MR3846204} yield further cases of the conjecture;
see also their survey~\cite{MR4073887}.
Chao and Kaneta~\cite{MR1477225} treated the cases
$11\leq q\leq19$.
Other important results on arcs and the associated length bounds
include the work of Hirschfeld and Korchm\'aros~\cite{MR1398077},
Kaneta and Maruta~\cite{MR985681}, and Storme and
Thas~\cite{MR1198385}.
For codes obtained from algebraic curves, restrictions on the
underlying curve provide additional arithmetic information that
can be used to study the length problem.

Denote by $\mathbb{F}_q$ the finite field with $q$ elements.
Let $F/\Fq$ be a function field of genus $g$ with full
constant field $\Fq$, and let $P_1,\ldots,P_n$ be distinct
rational places of $F$. Put $D=P_1+\cdots+P_n$, and let $G$ be a
divisor satisfying
\[
    \operatorname{supp}(G)\cap\operatorname{supp}(D)=\varnothing,
    \qquad 2g-2<\deg G<n.
\]
The algebraic geometry code $C_\mathcal{ L}(D,G)$ is obtained by evaluating
the functions in the Riemann--Roch space $\mathcal{ L}(G)$ at these
places. It has dimension $k=\deg G-g+1$ and minimum distance
$d\geq n-\deg G$ (see \cite{MR2464941}).

The MDS main conjecture for algebraic geometry codes has also
been investigated. Munuera~\cite{MR1178199} introduced an approach that
established cases of the conjecture for curves of genera one
and two over sufficiently large fields.
Further developments were obtained by Chen~\cite{MR1289576}
and Walker~\cite{MR1412183}.
Chen and Yau~\cite{MR1454971} studied geometric hyperelliptic
MDS codes and proved the conjecture under explicit hypotheses
on the genus and the field size.
These results motivate a more specific question: how far below
the general $q+1$ benchmark can the maximal length lie when the
underlying curves belong to a fixed class?

When $g=1$, the resulting codes are called elliptic codes.
Let $\operatorname{MEC}(k,q)$ denote the maximal length of a
non-trivial $q$-ary MDS elliptic code of dimension $k$.
Munuera~\cite{MR1178199} obtained the bound
\[
    \operatorname{MEC}(k,q)
    \leq \frac{q+1}{2}+\sqrt q+k.
\]

Li, Wan, and Zhang~\cite{7282884} proved that, for $\epsilon>0$
and $q>4/\epsilon^2$, there is a constant $C_\epsilon>0$ such that
an elliptic code of length $n\geq(2/3+\epsilon)q$ is not MDS whenever
\[
    C_\epsilon\ln q<k<n-C_\epsilon\ln q.
\]
Both restrictions on the dimension are part of the statement.
They also conjectured that the maximal length of non-trivial MDS
elliptic codes is at most $(1/2+\epsilon)q$ for sufficiently large $q$.

Subsequently, Han and Ren~\cite{MR4564628} established the explicit bound
\[
    \operatorname{MEC}(k,q)\leq\frac{q+1}{2}+\sqrt q
\]
whenever $q\geq289$ and
$3\leq k\leq(q+1-2\sqrt q)/10$.
Their bound is attained when $q$ is an odd square and $k$ is odd
in this range. 

Their subsequent work~\cite[Theorem~1.2]{MR4756098} proved the
asymptotic conjecture of Li, Wan, and Zhang. More precisely, it gives
dimension-dependent bounds that imply the uniform estimate
\[
    n\leq\frac{|E(\mathbb F_q)|+6}{2}
      \leq\frac q2+\sqrt q+\frac72
\]
for non-trivial MDS codes from an elliptic curve $E/\mathbb F_q$.
Two recent preprints give further results for elliptic codes.
Chen, Hu, Huang, and Zhao~\cite{ChenHuHuangZhao2026}
determine maximal lengths for $q\geq289$ and
$3\leq k\leq(q+1-2\sqrt q)/10$.
Their results distinguish divisors supported on rational places from
divisors that may include places of higher degree.
Chen and Huang~\cite[Theorem~32]{ChenHuang2026} prove
$n\leq |E(\mathbb F_q)|/2$ for $q\geq47089$ and $3\leq k\leq n-3$.
These results concern genus one.

For genus two, Munuera~\cite{MR1178199} and Chen and
Yau~\cite{MR1454971} obtained the $q+1$ bound under field-size
hypotheses. Their method uses a divisor class whose representatives
avoid prescribed evaluation points.
Recent work by Zhu et al.~\cite{ZhuChenPengZhangZhao2026}
constructs MDS codes using split Jacobians of genus-two curves.
Here we seek a uniform upper bound for all genus-two curves in the
dimension range below.

We ask whether genus-two MDS codes satisfy a bound
$n\leq cq+O(\sqrt q)$ for an explicit absolute constant $c<1$.
The bound should hold for every genus-two function field and every
dimension $3\leq k\leq n-3$, in every characteristic.




In this paper, we give an affirmative answer with $c=2/3$.
Let $F/\Fq$ be a function field of genus $2$, and retain
the notation $C_\mathcal{ L}(D,G)$ introduced above. Our main result shows
that if $C_\mathcal{ L}(D,G)$ is MDS with
\[
    \deg G=k+1<n,\qquad 3\leq k\leq n-3,
\]
then
\begin{equation}
    n<\frac{2q+10\sqrt q+16}{3}.
    \label{eq:introduction-main-bound}
\end{equation}
This bound holds for every prime power $q$ and both parities
of $k$. The divisor $G$ may have places of degree greater than
one in its support.

Therefore, the maximal length of these codes is at most
$2q/3+O(\sqrt q)$, uniformly over the underlying genus-two
function field and the entire indicated dimension range.
Thus the leading coefficient in the length bound is reduced
from $1$ to $2/3$. For every prime power $q\geq125$, the bound
also gives a strict improvement over the $q+1$ benchmark.

The paper is organized as follows. We begin in Section~\ref{sec:Preliminary}
with the necessary preliminaries on function fields and
algebraic geometry codes. The main length bound is proved
in Section~\ref{sec:main-length-bound}. 
Section~\ref{sec:uniform-value-set} proves the key auxiliary
theorem used in the proof of the main length bound.

\section{Preliminaries}
\label{sec:Preliminary}
In this section, we introduce some standard results about algebraic geometry codes and the MDS property.

    \subsection{Algebraic Geometry Codes}

Let $\mathbb{F}_q$ denote the finite field with $q$ elements and let $F/\Fq$ be an algebraic function
field of genus $g$ with full constant field $\Fq$.
We use the terminology of \cite{MR2464941}.
Let $\mathbb P_F$ be the set of places of $F$, and put
\[
    \mathbb P_F^{(1)}
    =\{P\in\mathbb P_F:\deg P=1\},
    \qquad N_F=|\mathbb P_F^{(1)}|.
\]
A divisor of $F$ is a formal sum
\[
    A=\sum_{P\in\mathbb P_F} n_P P,
    \qquad n_P\in\mathbb Z,
\]
where all but finitely many coefficients $n_P$ are zero.

Here $\deg P$ is the degree of the residue field of $P$ over $\mathbb F_q$.
We write $\deg A=\sum_P n_P\deg P$ and
$\operatorname{supp}(A)=\{P:n_P\neq0\}$.
The notation $A\geq0$ means that every coefficient is nonnegative;
$A\leq B$ means $B-A\geq0$.
For a divisor $A$, write
\[
    \mathcal L(A)
    =\{0\}\cup\{f\in F^\times:(f)+A\geq0\},
    \qquad
    \ell(A)=\dim_{\Fq}\mathcal L(A).
\]
For a nonzero function $h\in F^\times$, its principal divisor
is defined by
\[
    (h)=\sum_{P\in\mathbb P_F}v_P(h)P.
\]
Two divisors $A$ and $A'$ are said to be linearly equivalent
if  $A'=A+(h)$
for some $h\in F^\times$.
Here $v_P$ denotes the normalized valuation at $P$. We write $(h)_0$ and $(h)_\infty$ for the effective zero and
pole divisors, so $(h)=(h)_0-(h)_\infty$.
We write $A\sim A'$ for this relation and $[A]$ for the
divisor class of $A$, consisting of all divisors linearly
equivalent to $A$. In particular, if $A'=A+(h)$ with
$h\in F^\times$, then
\begin{equation}
    \mathcal L(A')=h^{-1}\mathcal L(A),
    \qquad \ell(A')=\ell(A).
    \label{eq:linear-equivalence-spaces}
\end{equation}

Let $W$ be a canonical divisor. The Riemann--Roch theorem
\cite[Theorem~1.5.15]{MR2464941} gives
\begin{equation}
    \ell(A)=\deg A+1-g+\ell(W-A),
    \qquad \deg W=2g-2.
    \label{eq:genus-two-riemann-roch}
\end{equation}
We will also use the Hasse--Weil bound
\cite[Theorem~5.2.3]{MR2464941} in the form
\begin{equation}
  q+1-2g\sqrt q \leq N_F\leq q+1+2g\sqrt q.
    \label{eq:genus-two-hasse-weil}
\end{equation}

In what follows, we always assume that the genus $g$ of $F$ is $2$.
Let $P_1,\ldots,P_n\in\mathbb P_F^{(1)}$ be distinct rational places, and set
\[
    S=\{P_1,\ldots,P_n\},
    \qquad D=P_1+\cdots+P_n.
\]
Let $G$ be a divisor with
\begin{equation}
    \operatorname{supp}(G)\cap S=\varnothing,
    \qquad \deg G=k+1<n,
    \qquad 3\leq k\leq n-3.
    \label{eq:code-assumptions}
\end{equation}
The algebraic geometry code associated with $D$ and $G$ is
\[
    C_\mathcal{ L}(D,G)
    =\{(f(P_1),\ldots,f(P_n)):f\in\mathcal L(G)\}
    \subseteq\Fq^n.
\]
Since $\deg(G-D)<0$, the kernel of the evaluation map is zero.
Moreover, \eqref{eq:genus-two-riemann-roch} gives $\ell(G)=k$.
Thus $C_\mathcal{ L}(D,G)$ has dimension $k$, in agreement with the
standard parameter formula for algebraic geometry codes
\cite[Theorem~2.2.2 and Corollary~2.2.3]{MR2464941}.

\subsection{A Characterization of the MDS Property and Dual Codes}

The following proposition expresses the usual information-set
characterization of an MDS code in terms of Riemann--Roch spaces
(see \cite[Lemma~2.2]{MR1454971}).

\begin{proposition}
\label{prop:mds-vanishing}
Under Condition~\eqref{eq:code-assumptions}, the code $C_\mathcal{ L}(D,G)$ is MDS
if and only if
\begin{equation}
    \ell(G-B)=0
    \quad\text{for every }0\leq B\leq D
    \text{ with }\deg B=k.
    \label{eq:mds-k-zeros}
\end{equation}
If the code is MDS, then, for every $0\leq j\leq k$ and every
choice of distinct rational places $Q_1,\ldots,Q_j\in S$,
\begin{equation}
    \ell(G-Q_1-\cdots-Q_j)=k-j.
    \label{eq:mds-vanishing}
\end{equation}
\end{proposition}

\begin{proof}
A nonzero function in $\mathcal L(G-B)$ gives a nonzero
codeword vanishing at the $k$ coordinates in
$\operatorname{supp}(B)$, and hence having weight at most $n-k$.
Conversely, any nonzero codeword of weight at most $n-k$
vanishes at some $k$ coordinates and gives such a function.
Therefore \eqref{eq:mds-k-zeros} is equivalent to
$d(C_\mathcal{ L}(D,G))\geq n-k+1$. The Singleton bound then gives equality.

Now assume that the code is MDS. Evaluation at any $k$
coordinates is an injective map from the $k$-dimensional space
$\mathcal L(G)$ to $\Fq^k$, and is therefore an
isomorphism. Extend $\{Q_1,\ldots,Q_j\}$ to a set of $k$
evaluation places. Projection onto the first $j$ coordinates
shows that evaluation at $Q_1,\ldots,Q_j$ is surjective.
Its kernel is $\mathcal L(G-Q_1-\cdots-Q_j)$, so
\eqref{eq:mds-vanishing} follows from the rank--nullity formula.
\end{proof}



The dual code of an $[n,k,n-k+1]_q$ MDS code is an
$[n,n-k,k+1]_q$ MDS code \cite[Ch.~11]{MR465509}.
For evaluation codes, choose a Weil differential $\eta$ satisfying
\[
    v_{P_i}(\eta)=-1,
    \qquad \operatorname{res}_{P_i}(\eta)=1
    \quad(1\leq i\leq n),
\]
whose existence follows from \cite[Lemma~2.2.9]{MR2464941}.
By the duality formula
\cite[Theorem~2.2.8 and Proposition~2.2.10]{MR2464941},
\begin{equation}
    C_\mathcal{ L}(D,G)^\perp=C_\mathcal{ L}(D,G^\perp),
    \qquad G^\perp=D-G+(\eta).
    \label{eq:ag-duality}
\end{equation}
The coefficient of $G^\perp$ at every $P_i$ is zero, and
\[
    \deg G^\perp=n-k+1=(n-k)+1.
\]
Consequently, the dual code satisfies
\eqref{eq:code-assumptions}, with dimension $n-k$.
Replacing the code by its dual code if necessary, we may assume
without loss of generality that
\begin{equation}
    3\leq k\leq n/2.
    \label{eq:dimension-reduction}
\end{equation}

\subsection{The Hyperelliptic Involution and Degree-Two Divisors}
For a finite separable extension $L'/L$ and places $Q\mid P$, write
$e(Q\mid P)$ for the ramification index and $f(Q\mid P)$ for the
residue field degree. The conorm of $P$ is
\[
    \operatorname{Con}_{L'/L}(P)=\sum_{Q\mid P}e(Q\mid P)Q.
\]
Conorm extends linearly to divisors.

A function field of genus at least $2$ is called
hyperelliptic if it is a quadratic extension of
a rational function field over the same constant field.
Every genus-two function field is hyperelliptic
\cite[Lemma~6.2.2]{MR2464941}. 
Fix $x\in F$ such that $[F:\Fq(x)]=2$, and put $H=(x)_\infty$.
The extension $F/\Fq(x)$ is Galois; see
\cite[Sec.~6.2]{MR2464941}. Write
\[
    \operatorname{Gal}(F/\Fq(x))=\{1,\sigma\}.
\]
The automorphism $\sigma$ is called the hyperelliptic involution.
Two distinct rational places $P$ and $Q$ form a
\emph{conjugate pair} if $Q=\sigma(P)$.
We collect some basic properties of the divisor $H$.
These properties will allow us to replace conjugate pairs
of rational places without changing the divisor class.
\begin{lemma}
\label{lem:hyperelliptic-divisors}
The divisor $H$ has degree $2$, is linearly equivalent to $W$,
and satisfies
\[
    \mathcal L(H)=\langle1,x\rangle_{\Fq}.
\]
Moreover, the following statements hold.
\begin{enumerate}
    \item For every $P\in\mathbb P_F^{(1)}$,
    \begin{equation}
        P+\sigma(P)\sim H.
        \label{eq:conjugate-pair}
    \end{equation}
    If $\sigma(P)=P$, the left-hand side means $2P$.

    \item For every divisor $B$ of degree $2$,
    \begin{equation}
        \ell(B)=
        \begin{cases}
            2,&B\sim H,\\
            1,&B\not\sim H.
        \end{cases}
        \label{eq:degree-two-dimension}
    \end{equation}
    In particular, distinct effective divisors of degree $2$
    can be linearly equivalent only if both are equivalent to
    $H$.

    \item If $z\in F$ satisfies $[F:\Fq(z)]=2$, then
    \begin{equation}
        z=\frac{c+dx}{a+bx},
        \qquad
        a,b,c,d\in\Fq,\quad ad-bc\neq0.
        \label{eq:quadratic-generator}
    \end{equation}
    In particular, $\Fq(z)=\Fq(x)$.
\end{enumerate}
\end{lemma}

\begin{proof}
For completeness, we provide a proof.
The degree formula for a nonconstant function
\cite[Theorem~1.4.11]{MR2464941} gives
$\deg H=[F:\Fq(x)]=2$.
Since $1,x\in\mathcal L(H)$ are linearly independent,
\eqref{eq:genus-two-riemann-roch} implies
\[
    2\leq\ell(H)=1+\ell(W-H).
\]
A divisor of degree zero has a nonzero Riemann--Roch space
only if it is principal \cite[Corollary~1.4.12]{MR2464941}. Thus $H\sim W$, and $\ell(H)=2$.

For (1), let $R$ be the restriction of $P$ to
$\Fq(x)$. Since $P$ is rational, so is $R$.
Its conorm in $F$ is $P+\sigma(P)$, including multiplicity
when $P$ is ramified. The divisor $H$ is the conorm of the pole
of $x$ in $\Fq(x)$. Any two rational places of a
rational function field are linearly equivalent, so
\eqref{eq:conjugate-pair} follows.

For (2), Riemann--Roch gives
\[
    \ell(B)=1+\ell(W-B).
\]
Since $\deg(W-B)=0$, the last term is $1$ precisely when
$B\sim W\sim H$, and is $0$ otherwise.
Thus,
\begin{equation}
        \ell(B)=
        \begin{cases}
            2,&B\sim H,\\
            1,&B\not\sim H.
        \end{cases}
    \end{equation}
    
If $B_1\neq B_2$ are effective and linearly equivalent, a
nonconstant function $f$ with $B_1+(f)=B_2$, together with $1$,
belongs to $\mathcal L(B_1)$. Hence $\ell(B_i)\geq2$ and $ B_i \sim H$ for $i=1,2$, proving
the final assertion.

For (3), put $B=(z)_\infty$. Then $\deg B=2$, and
$1,z\in\mathcal L(B)$ are linearly independent. By (2),
$B\sim H$. Choose $h\in F^\times$ with $B=H+(h)$.
Equation~\eqref{eq:linear-equivalence-spaces} shows that
$h,hz$ form a basis of $\mathcal L(H)$. Thus
\[
    h=a+bx,\qquad hz=c+dx
\]
for an invertible coefficient matrix over $\Fq$.
Taking their quotient proves \eqref{eq:quadratic-generator}.
\end{proof}
A rational place $P$ of $F$ is called a Weierstrass point if $\sigma(P)=P$. The following lemma bounds the number of rational
Weierstrass points.
\begin{lemma}[{\cite[Lemma 3.1]{MR1454971}}]
\label{lem:fixed-rational-places}
Let
\[
    w=|\{P\in\mathbb P_F^{(1)}:\sigma(P)=P\}|.
\]
Then $w\leq6$.
\end{lemma}
\begin{proof}
A fixed rational place is ramified in $F/\mathbb F_q(x)$:
its decomposition group has order two and its residue degree is one.
The Hurwitz formula~\cite[Theorem~3.4.13]{MR2464941} gives
\[
    \deg\operatorname{Diff}(F/\mathbb F_q(x))
      =2g(F)-2-2(-2)=6.
\]
Each ramified rational place contributes at least one to this degree.
This proves the bound, including in characteristic two.
\end{proof}

\subsection{Degree-Three Divisors and Injective Functions}

For $t\in F\setminus\Fq$, define
the set of values of $t$ at the rational places of $F$ by
\begin{equation}
    V_F(t)=\{t(P):P\in\mathbb P_F^{(1)}\}
    \subseteq\Fq\cup\{\infty\},
    \label{eq:value-set-definition}
\end{equation}
where $t(P)=\infty$ if $P$ is a pole of $t$. We identify $\mathbb F_q\cup\{\infty\}$ with the rational
points $\mathbb P^1(\mathbb F_q)$ of the projective line.
The degree of $t$ means
$\deg(t)_\infty=[F:\Fq(t)]$.

Let $A$ be a divisor of degree $3$.
Then $\ell(A)=2$ by the Riemann--Roch Theorem. 
For any $\Fq$-basis $u,v$ of
$\mathcal L(A)$, the quotient $t=v/u$ is nonconstant by the linear independence of $u,v$.
We first determine the possible degrees of the
associated quotients $t$ and characterize the quadratic case.
\begin{lemma}
\label{lem:degree-three-quotient}
Let $A$ be a divisor of degree $3$.
 For any $\Fq$-basis $u,v$ of
$\mathcal L(A)$, the quotient $t=v/u$ satisfies
\[
    [F:\Fq(t)]\in\{2,3\},
\]
and this extension is separable.
Furthermore,
\begin{equation}
    [F:\Fq(t)]=2
    \quad\Longleftrightarrow\quad
    A\sim H+R_0
    \text{ for some }R_0\in\mathbb P_F^{(1)}.
    \label{eq:quadratic-quotient-criterion}
\end{equation}
Equivalently,
\begin{equation}
    [F:\Fq(t)]=3
    \quad\Longleftrightarrow\quad
    \ell(A-H)=0.
    \label{eq:cubic-quotient-criterion}
\end{equation}
\end{lemma}

\begin{proof}
Both $A+(u)$ and $A+(v)$ are effective, and
\[
    (t)+A+(u)=A+(v)\geq0.
\]
Consequently,
\begin{equation}
    (t)_\infty\leq A+(u),
    \label{eq:pole-divisor-bound}
\end{equation}
and $\deg(t)_\infty\leq3$.
Degree one would make $F=\Fq(t)$ rational, which is
impossible since $g(F)=2$.

To check separability, let $p=\operatorname{char}\Fq$.
If $F/\Fq(t)$ were inseparable, then $t=z^p$ for
some $z\in F$, by \cite[Proposition~3.10.2(d)]{MR2464941}.
It follows that
\[
    \deg(t)_\infty=p\deg(z)_\infty.
\]
As $\deg(t)_\infty$ is either $2$ or $3$, this forces
$\deg(z)_\infty=1$, again contradicting $g(F)=2$.

Suppose that $\deg(t)_\infty=2$.
Lemma~\ref{lem:hyperelliptic-divisors} gives
$(t)_\infty\sim H$.
By \eqref{eq:pole-divisor-bound}, the divisor
\[
    R_0=A+(u)-(t)_\infty
\]
is effective of degree one, and therefore is a rational place.
Hence $A\sim H+R_0$.

Conversely, suppose that $A=H+R_0+(h)$ for some
$h\in F^\times$. Then $hu,hv$ form a basis of
$\mathcal L(H+R_0)$. Since $R_0\geq0$,
\[
    \mathcal L(H)\subseteq\mathcal L(H+R_0).
\]
Both spaces have dimension $2$, so they are equal.
We may therefore write
\[
    hu=a+bx,\qquad hv=c+dx,
    \qquad a,b,c,d\in\Fq,\quad ad-bc\neq0.
\]
Thus
\[
    t=\frac{v}{u}=\frac{c+dx}{a+bx},
    \qquad \Fq(t)=\Fq(x),
\]
which proves \eqref{eq:quadratic-quotient-criterion}.
Finally, $\deg(A-H)=1$, so $\ell(A-H)>0$ holds precisely when
$A-H$ is linearly equivalent to a rational place.
This proves \eqref{eq:cubic-quotient-criterion}.
\end{proof}
We next show that some properties of $t=v/u$ used below are
independent of the choice of basis $\{u,v \}$ and of the representative
of the divisor class $[A]$.
\begin{lemma}
\label{lem:value-set-invariance}
For a divisor $A$ of degree $3$, the generated subfield $\Fq(v/u)$, the degree of $v/u$ and the
cardinality $|V_F(v/u)|$ depend only on $[A]$.
For any fixed subset $S_0\subseteq\mathbb P_F^{(1)}$, injectivity
of $v/u$ on $S_0$ also depends only on $[A]$.
\end{lemma}
\begin{proof}
Fix a divisor $A$ of degree $3$, and let $u,v$ and
$\widetilde u,\widetilde v$ be two $\Fq$-bases of
$\mathcal L(A)$.
There exist $a,b,c,d\in\Fq$ with $ad-bc\neq0$ such that
\[
    \widetilde u=au+bv,\qquad
    \widetilde v=cu+dv.
\]
Put   $t=\frac vu$ and $\widetilde t=\frac{\widetilde v}{\widetilde u}.$
Then we have $\widetilde t=\frac{c+dt}{a+bt}.$
Consequently,
\[
    \Fq(\widetilde t)=\Fq(t),
\]
so $t$ and $\widetilde t$ have the same degree.

Consider the fractional linear transformation
\[
    \varphi(z)=\frac{c+dz}{a+bz}.
\]
This transformation permutes $\mathbb P^1(\mathbb F_q)$, since it is invertible.
For every rational place $P$ of $F$, we have
\[
    \widetilde t(P)=\varphi(t(P)).
\]
It follows that
\[
    V_F(\widetilde t)=\varphi(V_F(t)),
    \qquad
    |V_F(\widetilde t)|=|V_F(t)|.
\]
Moreover, for any $P,Q\in S_0$,
\[
    \widetilde t(P)=\widetilde t(Q)
    \quad\Longleftrightarrow\quad
    t(P)=t(Q).
\]
Therefore $\widetilde t$ is injective on $S_0$ if and only
if $t$ is injective on $S_0$.

Finally, let $A'$ be another representative of $[A]$.
Write $A'=A+(h)$ for some $h\in F^\times$.
Since
\[
    \mathcal L(A')=h^{-1}\mathcal L(A),
\]
the functions $u/h,v/h$ form a basis of $\mathcal L(A')$,
and their quotient is
\[
    \frac{v/h}{u/h}=\frac vu=t.
\]
Any other basis of $\mathcal L(A')$ gives the same degree,
cardinality $|V_F(v/u)|$, and injectivity property on $S_0$ by the
argument above.
This completes the proof.
\end{proof}



By Lemma~\ref{lem:value-set-invariance}, we may define a
function on divisor classes of degree $3$ by
\[
    V([A])=|V_F(t)|,
    \qquad t=\frac vu,
\]
where $\{u,v\}$ is any $\mathbb F_q$-basis of $\mathcal L(A)$.
This definition is independent of both the representative
$A$ and the chosen basis.

The following lemma provides a criterion for $v/u$ to be
injective on a set $S_0$ of rational places.

\begin{lemma}[Injectivity criterion]
\label{lem:injectivity}
Let $A$ be a divisor of degree $3$, and let
$S_0\subseteq\mathbb P_F^{(1)}$ contain at least two places.
Suppose that
\begin{equation}
    \ell(A-P-Q)=0
    \quad\text{for all distinct }P,Q\in S_0.
    \label{eq:pairwise-vanishing}
\end{equation}
Then, for every basis $u,v$ of $\mathcal L(A)$, the quotient
$t=v/u$ is injective on $S_0$. In particular,
\begin{equation}
    |S_0|\leq V([A]).
    \label{eq:injection-value-set}
\end{equation}
No disjointness assumption on $\operatorname{supp}(A)$ and
$S_0$ is required.
\end{lemma}

\begin{proof}
By the Weak Approximation Theorem
\cite[Theorem~1.3.1]{MR2464941}, choose $h\in F^\times$ such that
\[
    v_P(h)=-v_P(A)\qquad \forall  P\in S_0,
\]
where $v_P(A)$ denotes the coefficient of $P$ in $A$.
Put
\[
    A'=A+(h),\qquad u'=u/h,\qquad v'=v/h.
\]
Then $\operatorname{supp}(A')\cap S_0=\varnothing$,
$u',v'$ form a basis of $\mathcal L(A')$, and $v'/u'=t$.
Moreover, linear equivalence preserves
\eqref{eq:pairwise-vanishing}.

For distinct $P,Q\in S_0$, consider the evaluation map
\[
    \mathcal L(A')\longrightarrow\Fq^2,
    \qquad f\longmapsto(f(P),f(Q)).
\]
Its kernel is $\mathcal L(A'-P-Q)=\{0\}$.
Since both spaces have dimension $2$, the map is an
isomorphism. Consequently,
\[
    u'(P)v'(Q)-u'(Q)v'(P)\neq0.
\]
Thus,
$v'(P)/u'(P)$ and $v'(Q)/u'(Q)$ are distinct, with a zero
denominator interpreted as $\infty$.
These ratios equal $t(P)$ and $t(Q)$, respectively.
Therefore, $t$ is injective on $S_0$. Inequality~\eqref{eq:injection-value-set} follows from
the injectivity of $t$.
\end{proof}

Following Munuera~\cite{MR1178199}, for a positive integer $m$,
we consider the set of divisors
\[
    \mathcal D_m(S)
    =\{B:0\leq B\leq D,\ \deg B=m\}.
\]
Thus each element of $\mathcal D_m(S)$ is a sum of $m$
distinct rational places in $S$. Let $S_0\subseteq S$ contain at least two
places. We consider the following hypothesis.

\smallskip
\noindent
$C[F,S,S_0,m]$: There exists a class in
$\mathcal D_m(S)/{\sim}$ such that, for any two distinct
places $P,Q\in S_0$, there is a representative
$B_{P,Q}\in\mathcal D_m(S)$ of that class whose support
contains neither $P$ nor $Q$.

\smallskip
When $S_0=S$, this is Munuera's condition~\cite{MR1178199}.
For a proper subset $S_0\subsetneq S$, only the pairs to be
avoided are restricted to $S_0$. The representatives may
still be supported on the full evaluation set $S$.

The following corollary plays a key role in the proof of the main result in Section~\ref{sec:main-length-bound}.

\begin{corollary}
\label{cor:fixed-class-injection}
Let $S_0\subseteq S$ contain at least two places.
Suppose that $C_{\mathcal L}(D,G)$ is MDS and that
$C[F,S,S_0,k-2]$ holds.
Choose a divisor $B\in\mathcal D_{k-2}(S)$ whose class
satisfies this condition.

Let $A_0$ be any divisor linearly equivalent to $G-B$.
Then $\deg A_0=3$, and for every $\Fq$-basis $\{u,v\}$
of $\mathcal L(A_0)$, the quotient $ t=\frac vu$
is injective on $S_0$.
In particular, if $S_0$ contains a pair $P,\sigma(P)$
with $P\neq\sigma(P)$, then
    $[F:\Fq(t)]=3.$
\end{corollary}


\begin{proof}
For distinct $P,Q\in S_0$, choose
$B_{P,Q}\in\mathcal D_{k-2}(S)$ with
$B_{P,Q}\sim B$ and $B_{P,Q}\leq D-P-Q$.
Since $A_0\sim G-B_{P,Q}$, Proposition~\ref{prop:mds-vanishing} gives
\[
    \ell(A_0-P-Q)
    =\ell(G-B_{P,Q}-P-Q)=0.
\]
Lemma~\ref{lem:injectivity} proves the first assertion.

By Lemma~\ref{lem:degree-three-quotient}, the degree of $t$
is either $2$ or $3$. In the quadratic case,
$\Fq(t)=\Fq(x)$ by Lemma~\ref{lem:hyperelliptic-divisors},
so $t(P)=t(\sigma(P))$ for every rational place $P$.
If $S_0$ contains a conjugate pair, this contradicts
injectivity. Hence $[F:\Fq(t)]=3$.
\end{proof}

\section{Main Result}
\label{sec:main-length-bound}

We retain the notation of Section~\ref{sec:Preliminary} and put
\begin{equation}
    B(q)=\frac{2q+10\sqrt q+10}{3}.
    \label{eq:Bq-definition}
\end{equation}
The main result is the following.

\begin{theorem}
\label{thm:main-length}
Let $F/\Fq$ be a function field of genus $2$ with full
constant field $\Fq$. Let
$D=P_1+\cdots+P_n$ be a sum of distinct rational places, and let
$G$ be a divisor satisfying
\[
    \operatorname{supp}(G)\cap\operatorname{supp}(D)=\varnothing,
    \qquad \deg G=k+1<n,
    \qquad 3\leq k\leq n-3.
\]
If $C_\mathcal{ L}(D,G)$ is MDS, then
\begin{equation}
    n<B(q)+2=\frac{2q+10\sqrt q+16}{3}.
    \label{eq:main-length-bound}
\end{equation}
Equivalently,
\begin{equation}
    n\leq
    \left\lfloor\frac{2q+10\sqrt q+16}{3}\right\rfloor.
    \label{eq:main-integer-bound}
\end{equation}
\end{theorem}
The two formulations are equivalent since $B(q)$ is
never an integer, as verified in the proof of
Corollary~\ref{cor:noncyclic-cubic-value-bound}.
The proof constructs a fixed degree-three divisor class whose
associated quotient $t$ is injective on at least $n-2$ evaluation
places. A uniform bound for the cardinality of the set $V_F(t)$ of this
quotient $t$ then yields \eqref{eq:main-length-bound}.

\subsection{The Auxiliary Value-Set Theorem}
Let $\overline{\mathbb F}_q$ be an algebraic closure of $\Fq$.
Write $\overline F=F\overline{\mathbb F}_q$ for the extension
of the constant field to $\overline{\mathbb F}_q$.
Let $\mathcal E_F$ be the set of degree-three divisor classes
$[A]$ with the following property: for a basis $u,v$ of
$\mathcal L(A)$ and $t=v/u$, one has
\[
    [F:\Fq(t)]=3
    \quad\text{and}\quad
    \overline F/\overline{\mathbb F}_q(t)
    \text{ is Galois}.
\]
The extension $ \overline F/\overline{\mathbb F}_q(t)$ is then cyclic of degree $3$.
These conditions are independent of the basis and of the
representative of $[A]$ by
Lemma~\ref{lem:value-set-invariance}.
We call the elements of $\mathcal E_F$ \emph{exceptional classes}
and refer to $\mathcal E_F$ as the \emph{exceptional set}.

 The following result supplies the estimate of $|V_F(t)|$ needed
for the proof of Theorem~\ref{thm:main-length}. It also bounds the number of divisor
classes that must be avoided.
\begin{theorem}
\label{thm:uniform-value-set}
For every genus-two function field $F/\Fq$ with full
constant field $\Fq$,
\[
    |\mathcal E_F|\leq20.
\]
In characteristic $2$, one has $|\mathcal E_F|\leq1$.
If $A$ is a divisor of degree $3$ with
$[A]\notin\mathcal E_F$, then for every basis $u,v$ of
$\mathcal L(A)$, the quotient $t=v/u$ satisfies
\begin{equation}
    |V_F(t)|<B(q).
    \label{eq:uniform-value-set-bound}
\end{equation}
\end{theorem}

The proof of Theorem~\ref{thm:uniform-value-set} is independent
of the coding assumptions and will be given in Section~\ref{sec:uniform-value-set}.
In the constructions below, the quotients $v/u$ will
have degree $3$.

\subsection{Counting Conjugate Pairs and Divisor Classes}

Let $S=\{P_1,\ldots,P_n\}$ as before, and let $r$ be the number of
unordered pairs
\[
    \{P,\sigma(P)\}\subseteq S,\qquad P\neq\sigma(P).
\]
Distinct such pairs are disjoint.
As in Lemma~\ref{lem:fixed-rational-places}, let $w$ be the
number of rational places of $F$ fixed by $\sigma$.
To construct the auxiliary divisors, we need enough complete
conjugate pairs in $S$. The following count is the same orbit-counting argument used
in~\cite[Proposition~3.3]{MR1454971}, written here for an arbitrary
evaluation set $S$.
\begin{lemma}
\label{lem:complete-pair-count}
With the notation above, one has
\begin{equation}
    r\geq n-\frac{N_F+w}{2}
    \geq n-\frac{q+7+4\sqrt q}{2}.
    \label{eq:complete-pair-count}
\end{equation}
\end{lemma}

\begin{proof}
The action of $\sigma$ on $\mathbb P_F^{(1)}$ has $w$ orbits
of size one and $(N_F-w)/2$ orbits of size two.
Exactly $r$ of the latter are contained in $S$.
Every other orbit of size two contributes at most one place
to $S$, and the fixed places contribute at most $w$.
Therefore
\[
    n-r\leq
    r+\left(\frac{N_F-w}{2}-r\right)+w
    =\frac{N_F+w}{2}.
\]
The first inequality holds.
The second inequality follows from $w\leq6$ and the
Hasse--Weil bound \eqref{eq:genus-two-hasse-weil}.
\end{proof}
We also need enough distinct divisor classes to avoid the
exceptional set $\mathcal E_F$. This leads us to count the
classes represented by one or two evaluation places.
\begin{lemma}
\label{lem:candidate-class-count}
Define
\[
    \mathcal C_1(S)=\{[R]:R\in S\},
    \qquad
    \mathcal C_2(S)=\{[R+T]:R,T\in S,\ R\neq T\}.
\]
Then
\begin{equation}
    |\mathcal C_1(S)|=n
    \label{eq:degree-one-class-count}
\end{equation}
and
\begin{equation}
    |\mathcal C_2(S)|
    =\binom n2-r+\delta_r
    \geq\binom n2-\left\lfloor\frac n2\right\rfloor,
    \label{eq:degree-two-class-count}
\end{equation}
where $\delta_r=1$ if $r>0$ and $\delta_r=0$ if $r=0$.
\end{lemma}

\begin{proof}
Distinct rational places cannot be linearly equivalent.
Indeed, for any distinct rational places $R$ and $T$, an equivalence $R\sim T$ would give
a nonconstant function $z$ with pole divisor $T$, of degree one.
The degree formula
\cite[Theorem~1.4.11]{MR2464941} would then imply that $[F:\Fq(z)]=1$, contradicting $g(F)=2$.
This proves \eqref{eq:degree-one-class-count}.

There are $\binom n2$ distinct effective divisors $R+T$
with $R,T\in S$ and $R\neq T$.
By Lemma~\ref{lem:hyperelliptic-divisors}, two different
such divisors can be linearly equivalent only in the class
$[H]$.

Moreover,
\[
    R+T\sim H
    \quad\Longleftrightarrow\quad
    T=\sigma(R).
\]

Thus precisely $r$ of the divisors under consideration belong
to $[H]$, and every other divisor represents a different class.
This gives the equality in \eqref{eq:degree-two-class-count}.
The lower bound follows from $r\leq\lfloor n/2\rfloor$.
\end{proof}

\subsection{Constructing Injective Functions on Some Evaluation Places}

We now use conjugate pairs contained in $S$ to verify
$C[F,S,S_0,k-2]$ for a large subset $S_0\subseteq S$.

Consider first the case of even dimension. With enough
conjugate pairs contained in $S$, we can construct a function $t=v/u$ of degree $3$
that is injective on all but two evaluation places.


\begin{proposition}[Even dimension]
\label{prop:even-fixed-quotient}
Suppose that $C_{\mathcal L}(D,G)$ is MDS, $k=2s\geq4$, and $r\geq s+2$.
For any distinct $R,T\in S$, put $S_0=S\setminus\{R,T\}$.
Then the following statements hold.
\begin{enumerate}
    \item The condition $C[F,S,S_0,k-2]$ holds with $[R+T+(s-2)H]$
    as the required divisor class.

    \item For every divisor $A_0$ satisfying
    \begin{equation}
        A_0\sim G-R-T-(s-2)H,
        \label{eq:even-candidate-class}
    \end{equation}
    one has $\deg A_0=3$, and the quotient $t=v/u$ of any
    $\Fq$-basis $\{u,v\}$ of $\mathcal L(A_0)$ is injective
    on $S_0$ and has degree $3$. In particular,
    \[
        n-2\leq V([A_0]).
    \]

    \item If, in addition,
    \begin{equation}
        \binom n2-\left\lfloor\frac n2\right\rfloor>20,
        \label{eq:even-class-hypothesis}
    \end{equation}
    then
    \begin{equation}
        n<B(q)+2.
        \label{eq:even-conditional-bound}
    \end{equation}
\end{enumerate}
\end{proposition}

\begin{proof}
\begin{enumerate}
    \item Fix distinct $R,T\in S$ and put
    $S_0=S\setminus\{R,T\}$.
    For distinct $P,Q\in S_0$, the four places $R,T,P,Q$
    meet at most four conjugate pairs contained in $S$.
    Since $r\geq s+2$, at least  $r-4\geq s-2$
    pairs avoid all four places. Choose $s-2$ such pairs
    and put
    \[
        U_{P,Q}
        =\sum_{i=1}^{s-2}\bigl(Q_i+\sigma(Q_i)\bigr).
    \]
    All places in this sum are distinct and lie in
    $S\setminus\{R,T,P,Q\}$. By \eqref{eq:conjugate-pair},
    \[
        U_{P,Q}\sim(s-2)H.
    \]
    Hence the divisor $B_{P,Q}=R+T+U_{P,Q}$ satisfies
    \[
        0\leq B_{P,Q}\leq D-P-Q,
        \qquad
        \deg B_{P,Q}=k-2,
    \]
    and
    \[
        B_{P,Q}\sim R+T+(s-2)H.
    \]
    Thus $C[F,S,S_0,k-2]$ holds with the stated class.

    \item Let $A_0$ satisfy \eqref{eq:even-candidate-class}.
    Then
    \[
        \deg A_0=(2s+1)-2-2(s-2)=3.
    \]
    By part~(1) and Corollary~\ref{cor:fixed-class-injection},
    the quotient $t=v/u$ of any $\Fq$-basis $\{u,v\}$ of
    $\mathcal L(A_0)$ is injective on $S_0$.
    Removing $R,T$ destroys at most two conjugate pairs,
    so $S_0$ contains at least $r-2\geq s\geq2$
    such pairs. The same corollary therefore gives
    $[F:\Fq(t)]=3$. Moreover,
    \begin{equation}
        n-2\leq |V_F(t)|=V([A_0]).
        \label{eq:even-injective-values}
    \end{equation}

    \item Now allow $R,T$ to vary. The map
    \[
        [R+T]\longmapsto[G-R-T-(s-2)H]
    \]
    is injective on $\mathcal C_2(S)$. Indeed, equality
    of two image classes implies equality of the
    corresponding classes $[R+T]$.

    By Lemma~\ref{lem:candidate-class-count} and
    \eqref{eq:even-class-hypothesis}, its image contains
    more than $20$ classes. Since $|\mathcal E_F|\leq20$,
    we may choose $R,T$ so that the corresponding class
    $[A_0]$ lies outside $\mathcal E_F$.
    For this choice, Theorem~\ref{thm:uniform-value-set}
    and \eqref{eq:even-injective-values} give
    \[
        n-2\leq V([A_0])<B(q).
    \]
    Hence $n<B(q)+2$.
\end{enumerate}
\end{proof}

In odd dimension, the construction requires setting aside
only one evaluation place. This gives a stronger bound under
the corresponding pair and class conditions.

\begin{proposition}[Odd dimension]
\label{prop:odd-fixed-quotient}
Suppose that $C_{\mathcal L}(D,G)$ is MDS, $k=2s+1\geq3$, and  $r\geq s+2.$
For any $R\in S$, put $S_0=S\setminus\{R\}$.
Then the following statements hold.
\begin{enumerate}
    \item The condition $C[F,S,S_0,k-2]$ holds with $[R+(s-1)H]$
    as the required divisor class.

    \item For every divisor $A_0$ satisfying
    \begin{equation}
        A_0\sim G-R-(s-1)H,
        \label{eq:odd-candidate-class}
    \end{equation}
    one has $\deg A_0=3$, and the quotient $t=v/u$ of any
    $\Fq$-basis $\{u,v\}$ of $\mathcal L(A_0)$ is injective
    on $S_0$ and has degree $3$. In particular,
    \[
        n-1\leq V([A_0]).
    \]

    \item If, in addition, $n>20$, then
    \begin{equation}
        n<B(q)+1.
        \label{eq:odd-conditional-bound}
    \end{equation}
\end{enumerate}
\end{proposition}
\begin{proof}
\begin{enumerate}
    \item Fix $R\in S$ and put $S_0=S\setminus\{R\}$.
    For distinct $P,Q\in S_0$, the three places $R,P,Q$
    meet at most three conjugate pairs contained in $S$.
    Since $r\geq s+2$, at least $  r-3\geq s-1$
    pairs avoid all three places. Choose $s-1$ such pairs
    and put
    \[
        U_{P,Q}
        =\sum_{i=1}^{s-1}\bigl(Q_i+\sigma(Q_i)\bigr).
    \]
    All places in this sum are distinct and lie in
    $S\setminus\{R,P,Q\}$. By \eqref{eq:conjugate-pair},
    \[
        U_{P,Q}\sim(s-1)H.
    \]
    Hence the divisor $B_{P,Q}=R+U_{P,Q}$ satisfies
    \[
        0\leq B_{P,Q}\leq D-P-Q,
        \qquad
        \deg B_{P,Q}=k-2,
    \]
    and
    \[
        B_{P,Q}\sim R+(s-1)H.
    \]
    Thus $C[F,S,S_0,k-2]$ holds with the stated class.

    \item Let $A_0$ satisfy \eqref{eq:odd-candidate-class}.
    Then
    \[
        \deg A_0=(2s+2)-1-2(s-1)=3.
    \]
    By part~(1) and Corollary~\ref{cor:fixed-class-injection},
    the quotient $t=v/u$ of any $\Fq$-basis $\{u,v\}$ of
    $\mathcal L(A_0)$ is injective on $S_0$.
    Removing $R$ destroys at most one conjugate pair,
    so $S_0$ contains at least
        $r-1\geq s+1\geq2$
    such pairs. The same corollary therefore gives
    $[F:\Fq(t)]=3$. Moreover,
    \begin{equation}
        n-1\leq |V_F(t)|=V([A_0]).
        \label{eq:odd-injective-values}
    \end{equation}

    \item Now allow $R$ to vary. The map
    \[
        [R]\longmapsto[G-R-(s-1)H]
    \]
    is injective on $\mathcal C_1(S)$.
    Indeed, equality of two image classes implies
    equality of the corresponding classes $[R]$.
    By Lemma~\ref{lem:candidate-class-count}, its image
    contains $n$ distinct classes.
    If $n>20$, then $|\mathcal E_F|\leq20$ allows us
    to choose $R$ so that the corresponding class
    $[A_0]$ lies outside $\mathcal E_F$.
    For this choice, Theorem~\ref{thm:uniform-value-set}
    and \eqref{eq:odd-injective-values} give
    \[
        n-1\leq V([A_0])<B(q).
    \]
    Hence $n<B(q)+1$.
\end{enumerate}
\end{proof}

\subsection{Completion of the Proof}

\begin{proof}[Proof of Theorem~\ref{thm:main-length}]
By the duality reduction in \eqref{eq:ag-duality} and
\eqref{eq:dimension-reduction}, we may replace the code
by its dual if necessary and assume
\[
    3\leq k\leq n/2.
\]
This replacement preserves the length and the hypotheses
of the theorem.

\smallskip
\noindent
\textit{Case 1: $q\geq9$.}
Suppose, for a contradiction, that
\begin{equation}
    n\geq B(q)+2
      =\frac{2q+10\sqrt q+16}{3}.
    \label{eq:contrary-length-assumption}
\end{equation}
By Lemma~\ref{lem:complete-pair-count} and $k\leq n/2$, we have
\begin{align}
    r-\left\lfloor\frac k2\right\rfloor-2
    &\geq \frac{3n}{4}
       -\frac{q+7+4\sqrt q}{2}-2 \notag\\
    &\geq \frac{\sqrt q-3}{2}
     \geq0.
    \label{eq:sufficient-pair-count}
\end{align}
Thus
\begin{equation}
    r\geq\left\lfloor\frac k2\right\rfloor+2.
    \label{eq:pair-threshold}
\end{equation}
Moreover, $q\geq9$ and
\eqref{eq:contrary-length-assumption} imply
\[
    n\geq\left\lceil B(9)+2\right\rceil
      =\left\lceil\frac{64}{3}\right\rceil
      =22.
\]
Consequently,
\[
    n>20,
    \qquad
    \binom n2-\left\lfloor\frac n2\right\rfloor>20.
\]

If $k=2s$ is even, then
\eqref{eq:pair-threshold} gives $r\geq s+2$.
Part~(3) of Proposition~\ref{prop:even-fixed-quotient}
therefore yields
\[
    n<B(q)+2,
\]
contrary to \eqref{eq:contrary-length-assumption}.

If $k=2s+1$ is odd, then
\eqref{eq:pair-threshold} again gives $r\geq s+2$.
Since $n>20$, part~(3) of
Proposition~\ref{prop:odd-fixed-quotient} yields
\[
    n<B(q)+1,
\]
which also contradicts
\eqref{eq:contrary-length-assumption}.
This proves \eqref{eq:main-length-bound} for $q\geq9$.

\smallskip
\noindent
\textit{Case 2: $q<9$.}
The remaining prime powers are
$q\in\{2,3,4,5,7,8\}$.
Since $S\subseteq\mathbb P_F^{(1)}$, the Hasse--Weil bound
\cite[Theorem~5.2.3]{MR2464941} gives
\[
    n\leq N_F
      \leq\left\lfloor q+1+4\sqrt q\right\rfloor.
\]
For these field sizes, we have
\[
\begin{array}{c@{\qquad}c@{\qquad}c}
    q
    & \left\lfloor q+1+4\sqrt q\right\rfloor
    & \left\lceil B(q)+2\right\rceil-1\\[2pt]
    \hline
    2 & 8  & 11\\
    3 & 10 & 13\\
    4 & 13 & 14\\
    5 & 14 & 16\\
    7 & 18 & 18\\
    8 & 20 & 20
\end{array}
\]
It follows that
\[
    n\leq\left\lceil B(q)+2\right\rceil-1
      <B(q)+2.
\]
Thus \eqref{eq:main-length-bound} holds in every remaining
case. The integer bound \eqref{eq:main-integer-bound}
follows from the integrality of $n$.
\end{proof}

\begin{remark}
The bound in \eqref{eq:main-length-bound} is mainly useful
for large field sizes $q$. For prime powers $q<125$, its
right-hand side exceeds $q+1$, so it gives no improvement
over this benchmark. For $q\geq125$, however, the integer
bound in \eqref{eq:main-integer-bound} yields $n\leq q$.
Moreover, as $q\to\infty$, the theorem gives
\[
    n\leq \frac{2}{3}q+O(\sqrt q),
\]
which improves the leading coefficient from $1$ to $2/3$.
\end{remark}

\section{Proof of Theorem~\ref{thm:uniform-value-set}}
\label{sec:uniform-value-set}

This section proves Theorem~\ref{thm:uniform-value-set}.
Throughout, $F/\Fq$ has genus $2$ and full constant
field $\Fq$, and we write
$\overline F=F\overline{\mathbb F}_q$.

For each divisor $A$ of degree $3$, let $t=v/u$, where
$\{u,v\}$ is any $\Fq$-basis of $\mathcal L(A)$.
Define
\[
    \widetilde{\mathcal E}_F
    =
    \left\{
        [A]:
        \begin{array}{l}
            \deg A=3,\quad [F:\Fq(t)]=3,\\
            \overline F/\overline{\mathbb F}_q(t)
            \text{ is not Galois}
        \end{array}
    \right\},
\]
and
\[
    \mathcal Q_F
    =
    \left\{
        [A]:\deg A=3,\quad [F:\Fq(t)]=2
    \right\}.
\]
Both sets are well defined, since the generated subfield
$\Fq(t)$ depends only on $[A]$, as shown in the proof of
Lemma~\ref{lem:value-set-invariance}.

We first study the functions associated with classes in
$\widetilde{\mathcal E}_F$.
We then bound the number of exceptional classes in
$\mathcal E_F$ and complete the proof by treating
the classes in $\mathcal Q_F$.
\subsection{The Normal Closure and Its Constant Field}
We write $C_m$ for a cyclic group of order $m$, $S_m$ for the
symmetric group on $m$ letters, and $A_m$ for its alternating subgroup.
Let $M$ be the normal closure of $F/\Fq(t)$.
We first determine $\operatorname{Gal}(M/\Fq(t))$
and the full constant field of $M$. These facts will be used in
the comparison of rational places in its intermediate fields.
\begin{lemma}
\label{lem:geometric-s3-closure}
Let $A$ be a divisor with $[A]\in\widetilde{\mathcal E}_F$,
and let $\{u,v\}$ be any $\Fq$-basis of $\mathcal L(A)$.
Put $t=v/u$ and $K=\Fq(t)$, and let $M$ be the normal
closure of $F/K$. Then
\[
    \operatorname{Gal}(M/K)\simeq S_3,
\]
and the full constant field of $M$ is $\Fq$.
In particular, if $E=M^{A_3}$,
then $E$ also has full constant field $\Fq$  and
\[
    [E:K]=2,\qquad [M:E]=3,\qquad [M:F]=2.
\]
\end{lemma}

\begin{proof}
The extension $F/K$ is separable. Indeed, if it were
inseparable, then the characteristic would be $3$ and
$t=z^3$ for some $z\in F$, by
\cite[Proposition~3.10.2(d)]{MR2464941}.
It would follow that
\[
    3=\deg(t)_\infty=3\deg(z)_\infty,
\]
which would make $F=\Fq(z)$ rational.

Choose a primitive element for $F/K$ and let $f(X)$ be its
minimal polynomial. The field $M$ is the splitting field of
the separable irreducible cubic polynomial $f(X)$.

Since the three roots of $f(X)$ generate $M$ over $K$, any
$K$-automorphism of $M$ fixing all three roots is the identity.
Moreover, for any two roots $\alpha$ and $\beta$, the
irreducibility of $f(X)$ gives a $K$-isomorphism
$K(\alpha)\to K(\beta)$ sending $\alpha$ to $\beta$.
Since $M/K$ is normal, this isomorphism extends to a
$K$-automorphism of $M$.

Consequently, $\operatorname{Gal}(M/K)$ acts faithfully and
transitively on the three roots.
Faithfulness identifies it with a subgroup of $S_3$, so its
order divides $6$.
Transitivity and the orbit--stabilizer theorem imply that
its order is divisible by $3$.
Thus its order is either $3$ or $6$.
A group of order $3$ is cyclic, whereas a subgroup of $S_3$
of order $6$ is $S_3$ itself.
Therefore, $\operatorname{Gal}(M/K)$ is isomorphic to
either $C_3$ or $S_3$.

Set $\overline K=\overline{\mathbb F}_q(t)$ and
$\overline M=M\overline{\mathbb F}_q$.
 By
\cite[Proposition~3.6.1(c)]{MR2464941}, constant field extension preserves the degree
$[\overline F:\overline K]=3$.
Thus $f(X)$ remains irreducible over $\overline K$, and
$\overline M$ is its splitting field over $\overline K$.
Since $\overline F/\overline K$ is not Galois, its normal
closure has Galois group $S_3$. Hence
\[
    [\overline M:\overline K]=6.
\]
Let $\mathbb F_{q^m}$ denote the full constant field of $M$.
By \cite[Proposition~3.6.6]{MR2464941}, we have
\[
    [M:K]=[\mathbb F_{q^m}:\Fq][\overline M:\overline K]=6m.
\]
On the other hand, $[M:K]\leq6$. Therefore $m=1$ and
$\operatorname{Gal}(M/K)\simeq S_3$.
The remaining assertions follow from the Galois correspondence
and $\Fq\subseteq E\subseteq M$.
\end{proof}

Lemma~\ref{lem:geometric-s3-closure} assumes that
$F/\Fq(t)$ remains non-Galois after extending the constant
field to $\overline{\Fq}$.
This assumption is essential for the conclusion that
$\Fq$ is the full constant field of $M$. A non-Galois cubic extension over
$\Fq(t)$ can become Galois after extending constants;
an example is given at the end of this section.

\subsection{Different Divisors and a Genus Bound}

We retain the fields $K\subseteq F\subseteq M$ and $E$ from
Lemma~\ref{lem:geometric-s3-closure}.
For a finite separable extension $L'/L$ of function fields,
the \emph{different} is the effective divisor
\[
    \operatorname{Diff}(L'/L)
       =\sum_{P\in\mathbb P_L}
          \sum_{P'\mid P} d(P'\mid P)\,P',
\]
where $d(P'\mid P)$ denotes the different exponent of
$P'$ over $P$, as defined in
\cite[Definition~3.4.3]{MR2464941}.
In particular, $\operatorname{Diff}(L'/L)$ is a divisor
of the upper field $L'$.
All different divisors below are understood in this sense.

\begin{proposition}
\label{prop:different-identity}
One has
\begin{equation}
    \deg\operatorname{Diff}(F/K)
    =\deg\operatorname{Diff}(E/K)
     +\frac12\deg\operatorname{Diff}(M/E).
    \label{eq:different-comparison-identity}
\end{equation}
Consequently,
\begin{equation}
    \deg\operatorname{Diff}(E/K)\leq8,
    \qquad g(E)\leq3.
    \label{eq:resolvent-genus-bound}
\end{equation}
These assertions hold in every characteristic.
\end{proposition}

\begin{proof}
Identify $G=\operatorname{Gal}(M/K)$ with $S_3$, and let
\[
    H_1=\langle(12)\rangle,\qquad
    H_2=\langle(13)\rangle,\qquad
    H_3=\langle(23)\rangle.
\]
Put $F_j=M^{H_j}$, choosing the identification so that $F_1=F$.

Fix a place $R$ of $M$, and write $P=R\cap K$. Let $\mathcal O_R=\{z\in M:v_R(z)\geq0\}$ be its valuation ring.
For $i\geq0$, let $G_i(R|P)$ be the $i$-th ramification
group of $R$ over $P$, namely
\[
    G_i(R|P)=
    \{\tau\in G:
      v_R(\tau(z)-z)\geq i+1
      \text{ for every }z\in\mathcal O_R\}.
\]
For a subgroup $U\leq G$, the corresponding $i$-th
ramification group for $M/M^U$ is $G_i(R|P)\cap U$.
This follows directly from the definition, since the
valuation and the valuation ring remain unchanged;
see \cite[Definition~3.8.4]{MR2464941}.

Every nonidentity element of $S_3$ belongs to exactly one of
$H_1,H_2,H_3,A_3$. Thus, for each $i\geq0$,
\begin{equation}
    |G_i(R|P)|-1
    =\sum_{j=1}^3
       \bigl(|G_i(R|P)\cap H_j|-1\bigr)
     +|G_i(R|P)\cap A_3|-1.
    \label{eq:ramification-group-partition}
\end{equation}
Summing over $i\geq0$ and applying Hilbert's different formula
\cite[Theorem~3.8.7]{MR2464941} yields
\[
    d(R\mid P)
    =\sum_{j=1}^3 d(R\mid R\cap F_j)
     +d(R\mid R\cap E).
\]
Since this holds at every place $R$, we obtain an identity
of divisors of $M$:
\begin{equation}
    \operatorname{Diff}(M/K)
    =\sum_{j=1}^3\operatorname{Diff}(M/F_j)
     +\operatorname{Diff}(M/E).
    \label{eq:different-divisor-partition}
\end{equation}

The subgroups $H_1,H_2,H_3$ are conjugate in $S_3$.
Hence, for every $i$ and $j$, there exists $\tau\in G$ such
that $\tau H_i\tau^{-1}=H_j$. By Galois correspondence,
\[
    \tau(F_i)=\tau(M^{H_i})=M^{H_j}=F_j.
\]
Thus the fields $F_1,F_2,F_3$ are mutually $K$-conjugate.

If $\tau(F_i)=F_j$ for $\tau\in G$, then
\[
    \tau\bigl(\operatorname{Diff}(M/F_i)\bigr)
    =\operatorname{Diff}(M/F_j).
\]
Indeed, conjugation identifies the corresponding
ramification groups and hence preserves the different
exponents. It also preserves place degrees, because $\tau$
fixes $\Fq$ and induces isomorphisms of residue class fields.
Therefore the three different divisors have the same degree.
Writing
\[
    a=\deg\operatorname{Diff}(M/F),\qquad
    b=\deg\operatorname{Diff}(M/E),
\]
we deduce from \eqref{eq:different-divisor-partition} that
\begin{equation}
    \deg\operatorname{Diff}(M/K)=3a+b.
    \label{eq:total-different-degree}
\end{equation}

Apply transitivity of the different
\cite[Corollary~3.4.12]{MR2464941} to the two towers
$K\subseteq E\subseteq M$ and $K\subseteq F\subseteq M$.
All fields have the same full constant field, so the degrees
of conorms are multiplied by the corresponding extension
degrees. We obtain
\begin{align}
    \deg\operatorname{Diff}(M/K)
      &=b+3\deg\operatorname{Diff}(E/K),
        \label{eq:different-tower-E}\\
    \deg\operatorname{Diff}(M/K)
      &=a+2\deg\operatorname{Diff}(F/K).
        \label{eq:different-tower-F}
\end{align}
Comparing \eqref{eq:total-different-degree} with
\eqref{eq:different-tower-E} gives
$a=\deg\operatorname{Diff}(E/K)$.
Comparison with \eqref{eq:different-tower-F} now gives
\eqref{eq:different-comparison-identity}.

By the Hurwitz genus formula
\cite[Theorem~3.4.13]{MR2464941},
\begin{equation}
    \deg\operatorname{Diff}(F/K)
    =2g(F)-2-3(2g(K)-2)=8.
    \label{eq:cubic-different-eight}
\end{equation}
Since $\operatorname{Diff}(M/E)$ is effective, we obtain
$\deg\operatorname{Diff}(E/K)\leq8$.
Applying Hurwitz once more, to $E/K$, gives
\[
    2g(E)-2=-4+\deg\operatorname{Diff}(E/K)\leq4.
\]
Hence $g(E)\leq3$.
\end{proof}

\subsection{Counting the Values of a Degree-Three Function
at Rational Places}
The rational places of $K=\Fq(t)$ are naturally indexed by
the elements of
\[
    \mathbb P^1(\Fq)=\Fq\cup\{\infty\}.
\]
For $\alpha\in\Fq$, let $\mathfrak p_\alpha$ denote the
unique zero of $t-\alpha$ in $K$, and let
$\mathfrak p_\infty$ denote the unique pole of $t$.
Thus each $\alpha\in\Fq\cup\{\infty\}$ corresponds to a
unique rational place $\mathfrak p_\alpha$ of $K$.

For every rational place $\mathfrak p$ of $K$, define
\[
    a_{\mathfrak p}
       =\bigl|\{P\in\mathbb P_F^{(1)}
                    :P\mid\mathfrak p\}\bigr|,
    \qquad
    b_{\mathfrak p}
       =\bigl|\{Q\in\mathbb P_E^{(1)}
                    :Q\mid\mathfrak p\}\bigr|.
\]
We claim that, for every
$\alpha\in\Fq\cup\{\infty\}$,
\begin{equation}
    \alpha\in V_F(t)
    \quad\Longleftrightarrow\quad
    a_{\mathfrak p_\alpha}>0.
    \label{eq:value-place-correspondence}
\end{equation}
Indeed, if $a_{\mathfrak p_\alpha}>0$, there is a rational
place $P$ of $F$ lying above $\mathfrak p_\alpha$.
If $\alpha\in\Fq$, then
\[
    v_P(t-\alpha)
       =e(P\mid\mathfrak p_\alpha)
          v_{\mathfrak p_\alpha}(t-\alpha)>0,
\]
so $t(P)=\alpha$. If $\alpha=\infty$, then
\[
    v_P(t)
       =e(P\mid\mathfrak p_\infty)
          v_{\mathfrak p_\infty}(t)<0,
\]
so $P$ is a pole of $t$ and $t(P)=\infty$.
Thus $\alpha\in V_F(t)$.

Conversely, suppose that $\alpha\in V_F(t)$. Then there is a
rational place $P$ of $F$ such that $t(P)=\alpha$.
If $\alpha\in\Fq$, then $v_P(t-\alpha)>0$, so the restriction
of $P$ to $K$ is $\mathfrak p_\alpha$. If
$\alpha=\infty$, then $v_P(t)<0$, so the restriction of $P$
to $K$ is $\mathfrak p_\infty$. Hence $P$ is counted by
$a_{\mathfrak p_\alpha}$, and therefore
$a_{\mathfrak p_\alpha}>0$.


The following lemma expresses $|V_F(t)|$ in terms of the
numbers of rational places of $F$ and $E$, with a bounded
correction arising from total ramification in $F/K$.
\begin{lemma}
\label{lem:cubic-value-count}
Let $A$ be a divisor with $[A]\in\widetilde{\mathcal E}_F$,
and retain the notation $t$, $K$, $M$, and $E$ from
Lemma~\ref{lem:geometric-s3-closure}.
Write
\[
    N_F=|\mathbb P_F^{(1)}|,
    \qquad N_E=|\mathbb P_E^{(1)}|.
\]
Let $\mathcal T$ be the set of rational places of $K$
that are totally ramified in $F/K$, and define
\[
    \varepsilon
    =
    \sum_{\mathfrak p\in\mathcal T}
    |\{Q\in\mathbb P_E^{(1)}:Q\mid\mathfrak p\}|.
\]
Then $0\leq\varepsilon\leq8$ and
\begin{equation}
    3|V_F(t)|=N_F-N_E+2(q+1)+\varepsilon.
    \label{eq:cubic-value-count}
\end{equation}
\end{lemma}

\begin{proof}
Recall that, for each rational place $\mathfrak p$ of $K$,
\[
    a_{\mathfrak p}
       =|\{P\in\mathbb P_F^{(1)}:P\mid\mathfrak p\}|,
    \qquad
    b_{\mathfrak p}
       =|\{Q\in\mathbb P_E^{(1)}:Q\mid\mathfrak p\}|
\]
count the rational places of $F$ and $E$ above
$\mathfrak p$, respectively. Write $\mathbf 1_{\{a_{\mathfrak p}>0\}}$ for the indicator that
equals one when $a_{\mathfrak p}>0$ and zero otherwise.

Since all constant fields are $\mathbb F_q$ and
$\deg\mathfrak p=1$, every place above $\mathfrak p$
has degree equal to its relative degree over $\mathfrak p$.
By \eqref{eq:value-place-correspondence},
\[
    \alpha\in V_F(t)
    \quad\Longleftrightarrow\quad
    a_{\mathfrak p_\alpha}>0.
\]

Define the local correction term by
\[
    \varepsilon_{\mathfrak p}
       =3\mathbf 1_{\{a_{\mathfrak p}>0\}}
        -a_{\mathfrak p}+b_{\mathfrak p}-2.
\]
Thus
\begin{equation}
    3\mathbf 1_{\{a_{\mathfrak p}>0\}}
       =a_{\mathfrak p}-b_{\mathfrak p}+2
        +\varepsilon_{\mathfrak p}.
    \label{eq:local-value-count}
\end{equation}
We will prove that $\varepsilon_{\mathfrak p}=0$ unless
$\mathfrak p\in\mathcal T$, in which case
$\varepsilon_{\mathfrak p}=b_{\mathfrak p}$.

The local possibilities are summarized in the following table.
For an unramified place, the first column lists the residue degrees;
in the ramified rows, it lists the conorm in $F$.
\begin{center}
\begin{tabular}{lccc}
\toprule
Type over $\mathfrak p$ & $a_{\mathfrak p}$ & $b_{\mathfrak p}$
& $\varepsilon_{\mathfrak p}$\\
\midrule
Unramified: $1+1+1$ & $3$ & $2$ & $0$\\
Unramified: $1+2$ & $1$ & $0$ & $0$\\
Unramified: $3$ & $0$ & $2$ & $0$\\
Ramified: $2P_1+P_2$ & $2$ & $1$ & $0$\\
Totally ramified: $3P$ & $1$ & $0,1,$ or $2$ & $b_{\mathfrak p}$\\
\bottomrule
\end{tabular}
\end{center}
We justify the table in each case, including wild ramification.

\smallskip
\noindent
\textit{Case 1: $\mathfrak p$ is unramified in $F/K$.}
By \cite[Corollary~3.9.3(b)]{MR2464941}, it is also
unramified in the normal closure $M/K$, and hence in $E/K$.
For a place $R$ of $M$ above $\mathfrak p$, put
\[
    Q=R\cap E,
\]
and denote the decomposition group of $R$ over
$\mathfrak p$ by $G_Z(R\mid\mathfrak p)$.
The inertia group $G_T(R\mid\mathfrak p)$ is trivial since $\mathfrak p$ is unramified in $M/K$.
By \cite[Theorem~3.8.2]{MR2464941},
\[
    |G_Z(R\mid\mathfrak p)|=f(R\mid\mathfrak p),
\]
and $G_Z(R\mid\mathfrak p)$ is cyclic, being isomorphic
to the Galois group of an extension of finite residue class fields.
The cyclic subgroups of $S_3$ have orders $1$, $2$, or $3$;
in the last case $G_Z(R\mid\mathfrak p)=A_3$.

Since $\operatorname{Gal}(M/E)=A_3$, the decomposition
group of $R$ over $Q$ is
\[
    G_Z(R\mid Q)=G_Z(R\mid\mathfrak p)\cap A_3.
\]
Applying the same theorem to $M/E$ gives
\[
    f(R\mid Q)=|G_Z(R\mid\mathfrak p)\cap A_3|.
\]
The multiplicativity of relative degrees therefore yields
\[
    f(Q\mid\mathfrak p)
       =\frac{f(R\mid\mathfrak p)}{f(R\mid Q)}
       =\frac{|G_Z(R\mid\mathfrak p)|}
              {|G_Z(R\mid\mathfrak p)\cap A_3|}.
\]

The fundamental equality
\cite[Theorem~3.1.11]{MR2464941} shows that the residue
degrees in the unramified cubic extension $F/K$
have one of the following three types.

\begin{enumerate}
    \item If the type is $1+1+1$, then $\mathfrak p$
    splits completely in $F/K$. It also splits completely
    in $M/K$, by \cite[Lemma~3.9.5]{MR2464941}.
    Hence $G_Z(R\mid\mathfrak p)=\{1\}$ and
    $f(Q\mid\mathfrak p)=1$.
    As $E/K$ is an unramified quadratic extension at
    $\mathfrak p$, the place $\mathfrak p$ splits into
    two rational places of $E$.
    Thus $(a_{\mathfrak p},b_{\mathfrak p})=(3,2)$.

    \item If the type is $1+2$, let $P$ be the place of $F$
    with $f(P\mid\mathfrak p)=2$, and choose $R$ above $P$.
    Then
    \[
        |G_Z(R\mid\mathfrak p)|
           =f(R\mid\mathfrak p)
           =2f(R\mid P).
    \]
    Thus $2$ divides $|G_Z(R\mid\mathfrak p)|$, so
    $|G_Z(R\mid\mathfrak p)|=2$.
    Such a subgroup is generated by a transposition,
    and therefore
    \[
        G_Z(R\mid\mathfrak p)\cap A_3=\{1\}.
    \]
    Consequently $f(Q\mid\mathfrak p)=2$.
    Hence $\mathfrak p$ is inert in $E/K$, with one
    place of degree $2$ above it.
    Thus $(a_{\mathfrak p},b_{\mathfrak p})=(1,0)$.

    \item If the type is $3$, let $P$ be the unique place
    of $F$ above $\mathfrak p$, and choose $R$ above $P$. Let $Q=R \cap E$.
    Then
    \[
        |G_Z(R\mid\mathfrak p)|
           =f(R\mid\mathfrak p)
           =3f(R\mid P).
    \]
    Thus $|G_Z(R\mid\mathfrak p)|=3$ and
    $G_Z(R\mid\mathfrak p)=A_3$.
    Consequently $f(Q\mid\mathfrak p)=1$, so
    $\mathfrak p$ splits into two rational places of $E$.
    Thus $(a_{\mathfrak p},b_{\mathfrak p})=(0,2)$.
\end{enumerate}

In each of these three cases,
\[
    3\mathbf 1_{\{a_{\mathfrak p}>0\}}
       =a_{\mathfrak p}-b_{\mathfrak p}+2,
\]
and hence $\varepsilon_{\mathfrak p}=0$.

\smallskip
\noindent
\textit{Case 2: $\mathfrak p$ is ramified but not totally
ramified in $F/K$.}
Since $[F:K]=3$, the fundamental equality implies
\[
    \operatorname{Con}_{F/K}(\mathfrak p)=2P_1+P_2
\]
for some rational places $P_1$ and $P_2$ of $F$.
In particular, $a_{\mathfrak p}=2$.
Choose a place $R$ of $M$ above $P_1$, denote its inertia
group over $\mathfrak p$ by $G_T(R\mid\mathfrak p)$,
and put $Q=R\cap E$.
By multiplicativity of ramification indices,
\[
    e(R\mid\mathfrak p)=2e(R\mid P_1).
\]
Because $[M:F]=2$, the right-hand side is either $2$ or $4$.
On the other hand, $M/K$ is Galois of degree $6$, so
\[
    e(R\mid\mathfrak p)=|G_T(R\mid\mathfrak p)|
\]
divides $6$.
Therefore $|G_T(R\mid\mathfrak p)|=2$, and
\[
    G_T(R\mid\mathfrak p)\cap A_3=\{1\}.
\]

The inertia group of $R$ over $Q$ in $M/E$ is
\[
    G_T(R\mid Q)=G_T(R\mid\mathfrak p)\cap A_3.
\]
By \cite[Theorem~3.8.2]{MR2464941} and multiplicativity,
\[
    e(Q\mid\mathfrak p)
       =\frac{e(R\mid\mathfrak p)}{e(R\mid Q)}
       =\frac{|G_T(R\mid\mathfrak p)|}
              {|G_T(R\mid\mathfrak p)\cap A_3|}
       =2.
\]
Since $[E:K]=2$, this forces $Q$ to be the unique place
above $\mathfrak p$, with relative degree $1$.
Thus $b_{\mathfrak p}=1$ and
\[
    \varepsilon_{\mathfrak p}=3-2+1-2=0.
\]

\smallskip
\noindent
\textit{Case 3: $\mathfrak p$ is totally ramified in $F/K$.}
In this case,
\[
    \operatorname{Con}_{F/K}(\mathfrak p)=3P,
    \qquad \deg P=1,
\]
so $a_{\mathfrak p}=1$. Directly from the definition,
\[
    \varepsilon_{\mathfrak p}
       =3-1+b_{\mathfrak p}-2
       =b_{\mathfrak p}.
\]
The quadratic extension $E/K$ has at most two rational
places above $\mathfrak p$, whence
\[
    0\leq\varepsilon_{\mathfrak p}\leq2.
\]

These cases exhaust all possibilities. Hence
\[
    \sum_{\substack{\mathfrak p\in\mathbb P_K\\
                    \deg\mathfrak p=1}}
       \varepsilon_{\mathfrak p}
       =\sum_{\mathfrak p\in\mathcal T}b_{\mathfrak p}
       =\varepsilon.
\]
The Hurwitz genus formula
\cite[Theorem~3.4.13]{MR2464941} gives
\[
    \deg\operatorname{Diff}(F/K)
       =2g(F)-2-3(2g(K)-2)=8.
\]
For each $\mathfrak p\in\mathcal T$, its unique extension
$P$ in $F$ satisfies
\[
    d(P\mid\mathfrak p)\geq e(P\mid\mathfrak p)-1=2
\]
by \cite[Theorem~3.5.1]{MR2464941}.
Since $\deg P=1$, each such place contributes at least
$2$ to the degree of the different.
Distinct places in $\mathcal T$ give distinct places of $F$.
Therefore
\[
    2|\mathcal T|\leq8,
    \qquad
    0\leq\varepsilon\leq2|\mathcal T|\leq8.
\]
This argument applies to wild ramification as well.

Finally, sum \eqref{eq:local-value-count} over all $q+1$
rational places of $K$, including the infinite place.
Every rational place of $F$ or $E$ restricts to a rational
place of $K$, and each attained rational value of $t$
corresponds to exactly one such place. Thus
\[
    \sum_{\mathfrak p}a_{\mathfrak p}=N_F,
    \qquad
    \sum_{\mathfrak p}b_{\mathfrak p}=N_E,
    \qquad
    \sum_{\mathfrak p}\mathbf 1_{\{a_{\mathfrak p}>0\}}
       =|V_F(t)|.
\]
The resulting equality is
\[
    3|V_F(t)|=N_F-N_E+2(q+1)+\varepsilon,
\]
as claimed.
\end{proof}

\begin{corollary}
\label{cor:noncyclic-cubic-value-bound}
Let $A$ be a divisor with $[A]\in\widetilde{\mathcal E}_F$,
and let $t$ be as in Lemma~\ref{lem:geometric-s3-closure}.
Then
\begin{equation}
    |V_F(t)|<B(q).
    \label{eq:noncyclic-cubic-value-bound}
\end{equation}
\end{corollary}

\begin{proof}
By Lemma~\ref{lem:geometric-s3-closure} and
Proposition~\ref{prop:different-identity}, the field $E$
has full constant field $\Fq$ and genus at most $3$.
The Hasse--Weil bound
\cite[Theorem~5.2.3]{MR2464941} gives
\[
    N_F\leq q+1+4\sqrt q,
    \qquad
    N_E\geq q+1-6\sqrt q.
\]
Substituting these estimates into
\eqref{eq:cubic-value-count} yields
\[
    3|V_F(t)|
      \leq(q+1+4\sqrt q)-(q+1-6\sqrt q)
          +2(q+1)+8
      =2q+10\sqrt q+10.
\]
Thus $|V_F(t)|\leq B(q)$.

To obtain strict inequality, observe that $B(q)$ is never
an integer. If $q$ is not a square, then $B(q)$ is irrational.
If $q=a^2$ with $a$ an integer, then
\[
    3B(q)=2a^2+10a+10
      \equiv
      \begin{cases}
          1\pmod3,&a\equiv0,1\pmod3,\\
          2\pmod3,&a\equiv2\pmod3.
      \end{cases}
\]
The integrality of $|V_F(t)|$ proves
\eqref{eq:noncyclic-cubic-value-bound}.
\end{proof}

\subsection{The Cardinality of the Exceptional Set $\mathcal E_F$}

Recall the exceptional set $\mathcal E_F$ defined before
Theorem~\ref{thm:uniform-value-set}.
All automorphisms of $\overline F$ considered in this
subsection fix $\overline{\mathbb F}_q$ pointwise.

It remains to bound the number of exceptional
classes. We first associate each such class with a subgroup
of order $3$ in the automorphism group $\operatorname{Aut}_{\overline{\mathbb F}_q}(\overline F)$.
\begin{lemma}
\label{lem:exceptional-class-subgroups}
Distinct classes in $\mathcal E_F$ determine distinct
subgroups of order $3$ of
$\operatorname{Aut}_{\overline{\mathbb F}_q}(\overline F)$.
\end{lemma}

\begin{proof}
For $[A]\in\mathcal E_F$, choose a basis $u,v$ of
$\mathcal L(A)$ and put $t=v/u$.
Associate to $[A]$ the subgroup
\[
    T_A=\operatorname{Gal}
       (\overline F/\overline{\mathbb F}_q(t)).
\]
It has order $3$ and is independent of these choices by
Lemma~\ref{lem:value-set-invariance}.
Moreover, \eqref{eq:pole-divisor-bound} is an equality in
this case, since both divisors have degree $3$. Therefore
\begin{equation}
    (t)_\infty=A+(u),\qquad A\sim(t)_\infty.
    \label{eq:cubic-pole-class}
\end{equation}

Suppose that $T_A=T_{A'}$. We want to prove $[A]=[A']$. Let $t'$ be a corresponding
quotient for $A'$.
Since $T_A=T_{A'}$, their fixed fields in $\overline F$
coincide. By the Galois correspondence,
\[
    \overline{\mathbb F}_q(t)
      =\overline F^{T_A}
      =\overline F^{T_{A'}}
      =\overline{\mathbb F}_q(t').
\]
Thus, we can write
\[
    t'=\frac{at+b}{ct+d},
    \qquad
    a,b,c,d\in\overline{\Fq},
    \qquad ad-bc\neq0.
\]
If $c=0$, then $(t')_\infty=(t)_\infty$.
If $c\neq0$, then we have $(t')_\infty=(ct+d)_0$.
Since $(ct+d)_\infty=(t)_\infty$, we obtain
\[
    (t')_\infty-(t)_\infty
       =(ct+d)_0-(ct+d)_\infty
       =(ct+d).
\]
Thus the two pole divisors $ (t')_\infty$ and $(t)_\infty$ differ by a principal divisor
in $\overline F$ and are therefore linearly equivalent.
By \eqref{eq:cubic-pole-class}, the constant field extensions of
$A$ and $A'$ in $\overline{F}$ are linearly equivalent.

Constant field extension preserves Riemann--Roch dimensions
\cite[Theorem~3.6.3(d)]{MR2464941}.
Applying this to the degree-zero divisor $A-A'$ shows that
$\ell(A-A')=1$.
Hence $A-A'$ is already principal in $F$.
Thus $[A]=[A']$, proving injectivity. This completes the proof.
\end{proof}
The preceding injection reduces the problem to counting
subgroups of order $3$ of $\operatorname{Aut}_{\overline{\mathbb F}_q}(\overline F)$. The hyperelliptic structure gives
a uniform bound, with a stronger conclusion in characteristic $2$.
\begin{proposition}
\label{prop:exceptional-class-bound}
One has $|\mathcal E_F|\leq20$.
In characteristic $2$, one has the stronger bound
$|\mathcal E_F|\leq1$.
\end{proposition}

\begin{proof}
First suppose that the characteristic of $F$ is odd, and write
\[
    \mathcal A
      =\operatorname{Aut}_{\overline{\mathbb F}_q}(\overline F),
    \qquad
    \overline{\mathcal A}=\mathcal A/\langle\sigma\rangle.
\]
The hyperelliptic involution $\sigma$ is central in
$\mathcal A$, and restriction to the unique quadratic
rational subfield $\overline{\mathbb F}_q(x)$ gives an
embedding
\[
    \overline{\mathcal A}
      \hookrightarrow
      \operatorname{PGL}_2(\overline{\mathbb F}_q).
\]
The reduced automorphism group $\overline{\mathcal A}$ acts faithfully on
the six ramified places of
$\overline F/\overline{\mathbb F}_q(x)$;
see \cite[Secs.~2 and~3.1]{ShaskaVolklein2004}.
Indeed, all six ramified places are distinct in odd
characteristic, and a nonidentity fractional linear
transformation fixes at most two places of a rational
function field over an algebraically closed field.

A subgroup of order $3$ in $\mathcal A$ maps injectively to
$\overline{\mathcal A}$.
Furthermore, different such subgroups have different images.
To see this, a fixed reduced element has at most two lifts,
say $\tau$ and $\sigma\tau$. If $\tau^3=1$, centrality of $\sigma$ gives
\[
    (\sigma\tau)^3=\sigma\neq1.
\]
Thus a reduced element has at most one lift of order $3$.

An element of order $3$ in $\overline{\mathcal A}$ acts on
the six ramified places with cycle type $(3)(3)$.
The other non-trivial possibility, $(3)(1)(1)(1)$, would fix
three places and is impossible for a nonidentity fractional
linear transformation.
There are
\[
    \frac{1}{2}\binom{6}{3}=10
\]
ways to partition the six elements into two unordered
triples. The factor $1/2$ accounts for the fact that
interchanging the two triples gives the same partition.
Each triple supports exactly two distinct $3$-cycles.
Choosing one $3$-cycle on each triple therefore gives
$2\cdot2=4$ permutations for each partition.
Hence the number of permutations of cycle type $(3)(3)$ is
\[
    10\cdot2\cdot2=40.
\]
Every subgroup of order $3$ has two nonidentity elements,
so $\mathcal A$ has at most $20$ subgroups of order $3$.
Lemma~\ref{lem:exceptional-class-subgroups} now gives
$|\mathcal E_F|\leq20$.

Now suppose that the characteristic of $F$ is $2$.
The classification of $\operatorname{Aut}_{\overline{\mathbb F}_q}(\overline F)$ in
\cite{CardonaNartPujolas2005} gives the possibilities
\[
    C_2,\qquad C_2\times C_2,\qquad C_2\times S_3,
    \qquad M_{32},\qquad M_{160},
\]
where the last two groups have orders $32$ and $160$,
respectively.
Only $C_2\times S_3$ has order divisible by $3$, and it has
exactly one subgroup of order $3$, namely
$\{1\}\times A_3$.
Lemma~\ref{lem:exceptional-class-subgroups} therefore yields
$|\mathcal E_F|\leq1$.
This completes the proof.
\end{proof}

\subsection{Completion of the Proof}
\begin{proof}[Proof of Theorem~\ref{thm:uniform-value-set}]
The bounds on $|\mathcal E_F|$ follow from
Proposition~\ref{prop:exceptional-class-bound}.
Let $A$ be a divisor of degree $3$ with
$[A]\notin\mathcal E_F$.
Choose any $\Fq$-basis $\{u,v\}$ of $\mathcal L(A)$
and put $t=v/u$.
By Lemma~\ref{lem:degree-three-quotient} and the definitions
of the three sets, we have
\[
    [A]\in\widetilde{\mathcal E}_F\cup\mathcal Q_F.
\]
If $[A]\in\widetilde{\mathcal E}_F$, then
Corollary~\ref{cor:noncyclic-cubic-value-bound} gives
\[
    |V_F(t)|<B(q).
\]
If $[A]\in\mathcal Q_F$, then $[F:\Fq(t)]=2$, and
Lemma~\ref{lem:hyperelliptic-divisors} gives
    $\Fq(t)=\Fq(x).$
Every rational value attained at an unramified rational
place is attained at exactly two rational places,
interchanged by $\sigma$.
Every rational value attained at a ramified rational
place is attained at exactly one rational place,
fixed by $\sigma$.
Consequently,
\[
    |V_F(t)|=\frac{N_F+w}{2},
\]
where $w$ is the number of rational places fixed by $\sigma$.
Lemma~\ref{lem:fixed-rational-places} and the Hasse--Weil
bound yield
\[
    |V_F(t)|\leq\frac{q+7+4\sqrt q}{2}<B(q),
\]
where the strict inequality follows from
\[
    B(q)-\frac{q+7+4\sqrt q}{2}
    =\frac{q+8\sqrt q-1}{6}>0.
\]
Thus
\[
    V([A])=|V_F(t)|<B(q)
\]
for every divisor class $[A]$ of degree $3$ outside
$\mathcal E_F$ and every choice of basis.
This completes the proof.
\end{proof}



The following example shows that the bound $V([A])<B(q)$
can fail for classes in $\mathcal E_F$.

\begin{remark}[Necessity of the exceptional set]
The assumption that the cubic extension remains non-Galois
after extending the constant field to $\overline{\Fq}$
cannot be omitted from the value-set bound.

Let $q$ be odd with $q\equiv2\pmod3$, and choose
$a\in\Fq\setminus\{0,1\}$. Put
\[
    K=\Fq(t),\qquad
    h(t)=t(t-1)(t-a)^2,
\]
and define
\[
    F=K(z),\qquad z^3=h(t).
\]
For $\alpha\in\Fq$, let $\mathfrak p_\alpha$ denote
the zero of $t-\alpha$ in $K$, and let
$\mathfrak p_\infty$ denote the pole of $t$.
The principal divisor of $h(t)$ in $K$ is
\[
    (h(t))
      =\mathfrak p_0+\mathfrak p_1
       +2\mathfrak p_a-4\mathfrak p_\infty.
\]
Since $v_{\mathfrak p_0}(h(t))=1$, the element $h(t)$
is not a cube in $K$. Thus $X^3-h(t)$ is irreducible
over $K$. Since $\operatorname{char}\Fq\neq3$,
the extension $F/K$ is separable of degree $3$.

By \cite[Proposition~3.7.3(b) and Remark~3.7.5]{MR2464941},
exactly the four places
\[
    \mathfrak p_0,\quad \mathfrak p_1,\quad
    \mathfrak p_a,\quad \mathfrak p_\infty
\]
are ramified in $F/K$.
At each of these four places, the valuation of $h(t)$
is relatively prime to $3$. The ramification formula
therefore gives ramification index $3$.
The fundamental equality then implies that, for each
$\alpha\in\{0,1,a,\infty\}$, there is a unique place
$P_\alpha$ above $\mathfrak p_\alpha$, with
\[
    e(P_\alpha\mid\mathfrak p_\alpha)=3,
    \qquad
    f(P_\alpha\mid\mathfrak p_\alpha)=1.
\]
In particular, the residue class field at $P_0$ is $\Fq$,
so $\Fq$ is the full constant field of $F$.
All four places $P_\alpha$ are therefore rational.
The ramification is tame, and each different exponent
equals $2$. Hence
\[
    \operatorname{Diff}(F/K)
       =2(P_0+P_1+P_a+P_\infty).
\]
The Hurwitz genus formula now gives
\[
    2g(F)-2=3(-2)+8=2,
\]
so $g(F)=2$.

For $\alpha\in\Fq\setminus\{0,1,a\}$, the reduction of
$X^3-h(t)$ modulo $\mathfrak p_\alpha$ is
\[
    X^3-h(\alpha)\in\Fq[X].
\]
Since $\gcd(3,q-1)=1$, the cube map is a bijection
on $\Fq$. This polynomial has exactly one root in
$\Fq$ and is square-free. It therefore factors into
one linear factor and one irreducible quadratic factor.
By Kummer's theorem
\cite[Theorem~3.3.7]{MR2464941}, there are exactly
two places of $F$ above $\mathfrak p_\alpha$,
one rational and the other of degree $2$.

Together with the four totally ramified places,
this shows that every rational place of $K$ has
exactly one rational place of $F$ above it.
Every rational place of $F$ restricts to a rational
place of $K$. Consequently,
\[
    |V_F(t)|=N_F=q+1.
\]
For example,
\[
    B(125)=\frac{260+50\sqrt5}{3}<126.
\]
Thus the cubic value-set bound fails for this example
when $q=125$.

Finally, consider the constant field extension
\[
    \overline F=F\overline{\Fq}.
\]
The same valuation computation shows that $h(t)$
is not a cube in $\overline{\Fq}(t)$.
Since $\overline{\Fq}$ contains a primitive cube
root of unity, \cite[Proposition~3.7.3(a)]{MR2464941}
shows that
\[
    \overline F/\overline{\Fq}(t)
\]
is a cyclic Kummer extension of degree $3$.
Moreover,
\[
    A=(t)_\infty=3P_\infty
\]
has degree $3$. By the Riemann--Roch theorem,
$\ell(A)=2$, and hence
\[
    \mathcal L(A)=\langle1,t\rangle_{\Fq}.
\]
Therefore $[A]\in\mathcal E_F$, which explains why
this class must be excluded from the uniform bound.
\end{remark}

\section{Conclusion}
\label{sec:conclusion}

We established a uniform upper bound on the length of MDS
codes from genus-two curves. For an MDS $[n,k]_q$ code $C_\mathcal{ L}(D,G)$
with $\deg G=k+1$ and $3\leq k\leq n-3$, we proved that
\[
    n<\frac{2q+10\sqrt q+16}{3}.
\]
The bound holds for every prime power $q$ and both parities
of $k$. It gives a leading coefficient of $2/3$ and improves
the $q+1$ bound when $q\geq125$.

A central question is whether the coefficient $2/3$ is optimal.
Explicit constructions of long genus-two MDS codes would help
determine how closely this bound can be approached.
Sharper estimates for rational-place counts and ramification
contributions may also improve the terms of smaller order.
Another direction is to investigate whether similar methods
yield analogous length bounds for MDS codes from
hyperelliptic curves of higher genus.

\section*{Acknowledgement}
This work is supported by the National Natural Science Foundation of China (No. 12441107),   Guangdong Basic and Applied Basic Research Foundation (No. \seqsplit{2025A1515011764}), and the National Key Research and Development Program of
China (No.\seqsplit{2025YFA1017100}).

\bibliographystyle{ieeetr}
\bibliography{reference}

@misc{ChenHuHuangZhao2026,
      title={On the Maximal Length of MDS Elliptic Codes}, 
      author={Haojie Chen and Chuangqiang Hu and Junjie Huang and Chang-An Zhao},
      year={2026},
      eprint={2605.29439},
      archivePrefix={arXiv},
      primaryClass={cs.IT},
      url={https://arxiv.org/abs/2605.29439}, 
}

@misc{ChenHuang2026,
      title={Critical Numbers for Restricted Sumsets: Rigidity and Collapse in Finite Abelian Groups}, 
      author={Bocong Chen and Jing Huang},
      year={2026},
      eprint={2602.10402},
      archivePrefix={arXiv},
      primaryClass={math.CO},
      url={https://arxiv.org/abs/2602.10402}, 
}

@article {ZhuChenPengZhangZhao2026,
    AUTHOR = {Zhu, Yunlong and Chen, Hao and Peng, Fan and Zhang, Fangguo
              and Zhao, Chang-An},
     TITLE = {New explicit constructions of {MDS} and {NMDS} algebraic
              geometry codes},
   JOURNAL = {IEEE Trans. Inform. Theory},
  FJOURNAL = {Institute of Electrical and Electronics Engineers.
              Transactions on Information Theory},
    VOLUME = {72},
      YEAR = {2026},
    NUMBER = {9},
     PAGES = {6450--6465},
      ISSN = {0018-9448,1557-9654},
   MRCLASS = {94B27},
  MRNUMBER = {5119224},
       DOI = {10.1109/tit.2026.3712030},
       URL = {https://doi.org/10.1109/tit.2026.3712030},
}

@incollection{ShaskaVolklein2004,
  author       = {Shaska, Tanush and V{\"o}lklein, Helmut},
  title        = {Elliptic subfields and automorphisms of genus 2 function fields},
  booktitle    = {Algebra, Arithmetic and Geometry with Applications},
  editor       = {Christensen, Chris and Sathaye, Avinash and Sundaram, Ganesh
                  and Bajaj, Chandrajit},
  publisher    = {Springer},
  address      = {Berlin, Heidelberg},
  year         = {2004},
  pages        = {703--723},
  doi          = {10.1007/978-3-642-18487-1_42},
  eprint       = {math/0107142},
  archivePrefix = {arXiv},
  url          = {https://arxiv.org/abs/math/0107142}
}

@article{CardonaNartPujolas2005,
  author       = {Cardona, Gabriel and Nart, Enric and Pujol{\`a}s, Jordi},
  title        = {Curves of genus two over fields of even characteristic},
  journal      = {Mathematische Zeitschrift},
  volume       = {250},
  number       = {1},
  year         = {2005},
  pages        = {177--201},
  doi          = {10.1007/s00209-004-0750-0},
  eprint       = {math/0210105},
  archivePrefix = {arXiv},
  url          = {https://arxiv.org/abs/math/0210105}
}

@article {MR4564628,
    AUTHOR = {Han, Dongchun and Ren, Yuan},
     TITLE = {A tight upper bound for the maximal length of {MDS} elliptic
              codes},
   JOURNAL = {IEEE Trans. Inform. Theory},
  FJOURNAL = {Institute of Electrical and Electronics Engineers.
              Transactions on Information Theory},
    VOLUME = {69},
      YEAR = {2023},
    NUMBER = {2},
     PAGES = {819--822},
      ISSN = {0018-9448,1557-9654},
   MRCLASS = {94B05},
  MRNUMBER = {4564628},
MRREVIEWER = {Yan\ Liu},
       DOI = {10.1109/tit.2022.3211543},
       URL = {https://doi.org/10.1109/tit.2022.3211543},
}

@article {MR2911882,
    AUTHOR = {Ball, Simeon},
     TITLE = {On sets of vectors of a finite vector space in which every
              subset of basis size is a basis},
   JOURNAL = {J. Eur. Math. Soc. (JEMS)},
  FJOURNAL = {Journal of the European Mathematical Society (JEMS)},
    VOLUME = {14},
      YEAR = {2012},
    NUMBER = {3},
     PAGES = {733--748},
      ISSN = {1435-9855,1435-9863},
   MRCLASS = {15A03 (05B35 51E21 94B05)},
  MRNUMBER = {2911882},
MRREVIEWER = {Gengsheng\ Zhang},
       DOI = {10.4171/JEMS/316},
       URL = {https://doi.org/10.4171/JEMS/316},
}

@book {MR2464941,
    AUTHOR = {Stichtenoth, Henning},
     TITLE = {Algebraic function fields and codes},
    SERIES = {Graduate Texts in Mathematics},
    VOLUME = {254},
   EDITION = {Second},
 PUBLISHER = {Springer-Verlag, Berlin},
      YEAR = {2009},
     PAGES = {xiv+355},
      ISBN = {978-3-540-76877-7},
   MRCLASS = {14H05 (11R58 11T71 14G15 14G50 94B27)},
  MRNUMBER = {2464941},
}

@book {MR465509,
    AUTHOR = {MacWilliams, F. J. and Sloane, N. J. A.},
     TITLE = {The theory of error-correcting codes. {I}},
    SERIES = {North-Holland Mathematical Library},
    VOLUME = {Vol. 16},
 PUBLISHER = {North-Holland Publishing Co., Amsterdam-New York-Oxford},
      YEAR = {1977},
     PAGES = {i--xv and 1--369},
      ISBN = {0-444-85009-0},
   MRCLASS = {94A10},
  MRNUMBER = {465509},
MRREVIEWER = {Ian\ Blake},
}

@article {MR3846204,
    AUTHOR = {Ball, Simeon and Lavrauw, Michel},
     TITLE = {Planar arcs},
   JOURNAL = {J. Combin. Theory Ser. A},
  FJOURNAL = {Journal of Combinatorial Theory. Series A},
    VOLUME = {160},
      YEAR = {2018},
     PAGES = {261--287},
      ISSN = {0097-3165,1096-0899},
   MRCLASS = {51E21 (11G20 14H50)},
  MRNUMBER = {3846204},
MRREVIEWER = {Gy\"orgy\ Kiss},
       DOI = {10.1016/j.jcta.2018.06.015},
       URL = {https://doi.org/10.1016/j.jcta.2018.06.015},
}

@article {MR4073887,
    AUTHOR = {Ball, Simeon and Lavrauw, Michel},
     TITLE = {Arcs in finite projective spaces},
   JOURNAL = {EMS Surv. Math. Sci.},
  FJOURNAL = {EMS Surveys in Mathematical Sciences},
    VOLUME = {6},
      YEAR = {2019},
    NUMBER = {1-2},
     PAGES = {133--172},
      ISSN = {2308-2151,2308-216X},
   MRCLASS = {51E21 (05B25 94B05)},
  MRNUMBER = {4073887},
MRREVIEWER = {Yue\ Zhou},
       DOI = {10.4171/emss/33},
       URL = {https://doi.org/10.4171/emss/33},
}

@article {MR1477225,
    AUTHOR = {Chao, J. M. and Kaneta, H.},
     TITLE = {Classical arcs in {${\rm PG}(r,q)$} for {$11\leq q\leq19$}},
      NOTE = {Combinatorics (Rome and Montesilvano, 1994)},
   JOURNAL = {Discrete Math.},
  FJOURNAL = {Discrete Mathematics},
    VOLUME = {174},
      YEAR = {1997},
    NUMBER = {1-3},
     PAGES = {87--94},
      ISSN = {0012-365X,1872-681X},
   MRCLASS = {05B25},
  MRNUMBER = {1477225},
       DOI = {10.1016/S0012-365X(96)00319-6},
       URL = {https://doi.org/10.1016/S0012-365X(96)00319-6},
}

@article{1,
  title = {On the Embedding of an Arc into a Conic in a Finite Plane},
  author = {Hirschfeld, J.W.P. and Korchmáros, G.},
  year = {1996},
  journal = {Finite Fields and Their Applications},
  volume = {2},
  number = {3},
  pages = {274--292},
  issn = {1071-5797},
  doi = {10.1006/ffta.1996.0017},

}

@article {MR1398077,
    AUTHOR = {Hirschfeld, J. W. P. and Korchm\'aros, G.},
     TITLE = {On the embedding of an arc into a conic in a finite plane},
   JOURNAL = {Finite Fields Appl.},
  FJOURNAL = {Finite Fields and their Applications},
    VOLUME = {2},
      YEAR = {1996},
    NUMBER = {3},
     PAGES = {274--292},
      ISSN = {1071-5797,1090-2465},
   MRCLASS = {51E21},
  MRNUMBER = {1398077},
MRREVIEWER = {Tam\'as\ Sz\H onyi},
       DOI = {10.1006/ffta.1996.0017},
       URL = {https://doi.org/10.1006/ffta.1996.0017},
}

@article {MR985681,
    AUTHOR = {Kaneta, Hitoshi and Maruta, Tatsuya},
     TITLE = {An elementary proof and an extension of {T}has' theorem on
              {$k$}-arcs},
   JOURNAL = {Math. Proc. Cambridge Philos. Soc.},
  FJOURNAL = {Mathematical Proceedings of the Cambridge Philosophical
              Society},
    VOLUME = {105},
      YEAR = {1989},
    NUMBER = {3},
     PAGES = {459--462},
      ISSN = {0305-0041,1469-8064},
   MRCLASS = {51E20},
  MRNUMBER = {985681},
       DOI = {10.1017/S0305004100077823},
       URL = {https://doi.org/10.1017/S0305004100077823},
}

@article{3,
  title={{{MDS}} codes and arcs in PG (n, q) with q even: an improvement of the bounds of Bruen, Thas, and Blokhuis},
  author={Storme, Leo and Thas, Joseph A.},
  journal={Journal of Combinatorial Theory, Series A},
  volume={62},
  number={1},
  pages={139--154},
  year={1993},
  publisher={Elsevier}
}

@article {MR1198385,
    AUTHOR = {Storme, L. and Thas, J. A.},
     TITLE = {M.{D}.{S}.\ codes and arcs in {${\rm PG}(n,q)$} with {$q$}
              even: an improvement of the bounds of {B}ruen, {T}has, and
              {B}lokhuis},
   JOURNAL = {J. Combin. Theory Ser. A},
  FJOURNAL = {Journal of Combinatorial Theory. Series A},
    VOLUME = {62},
      YEAR = {1993},
    NUMBER = {1},
     PAGES = {139--154},
      ISSN = {0097-3165,1096-0899},
   MRCLASS = {51E21 (51E22 94B27 94B65)},
  MRNUMBER = {1198385},
MRREVIEWER = {J.\ W. P. Hirschfeld},
       DOI = {10.1016/0097-3165(93)90076-K},
       URL = {https://doi.org/10.1016/0097-3165(93)90076-K},
}

@article {MR1178199,
    AUTHOR = {Munuera, Carlos},
     TITLE = {On the main conjecture on geometric {MDS} codes},
   JOURNAL = {IEEE Trans. Inform. Theory},
  FJOURNAL = {Institute of Electrical and Electronics Engineers.
              Transactions on Information Theory},
    VOLUME = {38},
      YEAR = {1992},
    NUMBER = {5},
     PAGES = {1573--1577},
      ISSN = {0018-9448,1557-9654},
   MRCLASS = {94B05 (11T71 14H52 94B27)},
  MRNUMBER = {1178199},
MRREVIEWER = {Francesco\ Fabris},
       DOI = {10.1109/18.149510},
       URL = {https://doi.org/10.1109/18.149510},
}

@article {MR1412183,
    AUTHOR = {Walker, Judy L.},
     TITLE = {A new approach to the main conjecture on algebraic-geometric
              {MDS} codes},
      NOTE = {Second Upper Michigan Combinatorics Workshop on Designs, Codes
              and Geometries (Houghton, MI, 1994)},
   JOURNAL = {Des. Codes Cryptogr.},
  FJOURNAL = {Designs, Codes and Cryptography. An International Journal},
    VOLUME = {9},
      YEAR = {1996},
    NUMBER = {1},
     PAGES = {115--120},
      ISSN = {0925-1022,1573-7586},
   MRCLASS = {94B27 (11G20 11T71)},
  MRNUMBER = {1412183},
       DOI = {10.1023/A:1027358511882},
       URL = {https://doi.org/10.1023/A:1027358511882},
}

@article {MR1289576,
    AUTHOR = {Chen, Hao},
     TITLE = {On the main conjecture of geometric {MDS} codes},
   JOURNAL = {Internat. Math. Res. Notices},
  FJOURNAL = {International Mathematics Research Notices},
      YEAR = {1994},
    NUMBER = {8},
     PAGES = {313--318},
      ISSN = {1073-7928,1687-0247},
   MRCLASS = {94B27 (11T71 14G15)},
  MRNUMBER = {1289576},
MRREVIEWER = {Michael\ M.\ Dediu},
       DOI = {10.1155/S107379289400036X},
       URL = {https://doi.org/10.1155/S107379289400036X},
}

@article {MR1454971,
    AUTHOR = {Chen, Hao and Yau, Stephen S.-T.},
     TITLE = {Contribution to {M}unuera's problem on the main conjecture of
              geometric hyperelliptic {MDS} codes},
   JOURNAL = {IEEE Trans. Inform. Theory},
  FJOURNAL = {Institute of Electrical and Electronics Engineers.
              Transactions on Information Theory},
    VOLUME = {43},
      YEAR = {1997},
    NUMBER = {4},
     PAGES = {1349--1354},
      ISSN = {0018-9448,1557-9654},
   MRCLASS = {94B27},
  MRNUMBER = {1454971},
       DOI = {10.1109/18.605607},
       URL = {https://doi.org/10.1109/18.605607},
}

@INPROCEEDINGS{7282884,
  author={Li, Jiyou and Wan, Daqing and Zhang, Jun},
  booktitle={2015 IEEE International Symposium on Information Theory (ISIT)}, 
  title={On the minimum distance of elliptic curve codes}, 
  year={2015},
  volume={},
  number={},
  pages={2391-2395},
  doi={10.1109/ISIT.2015.7282884}}

@article {MR4756098,
    AUTHOR = {Han, Dongchun and Ren, Yuan},
     TITLE = {The maximal length of {$q$}-ary {MDS} elliptic codes is close
              to {$\frac q2$}},
   JOURNAL = {Int. Math. Res. Not. IMRN},
  FJOURNAL = {International Mathematics Research Notices. IMRN},
      YEAR = {2024},
    NUMBER = {11},
     PAGES = {9036--9043},
      ISSN = {1073-7928,1687-0247},
   MRCLASS = {14G50 (05E14 11B30 11T71 94B05)},
  MRNUMBER = {4756098},
MRREVIEWER = {Juan\ Tena},
       DOI = {10.1093/imrn/rnad271},
       URL = {https://doi.org/10.1093/imrn/rnad271},
}

\end{document}